\documentclass[journal]{IEEEtran}
\usepackage{multicol}

\usepackage{cite}

\ifCLASSINFOpdf
\usepackage[pdftex]{graphicx}
\DeclareGraphicsExtensions{.pdf,.jpeg,.png}
\else
\usepackage[dvips]{graphicx}
\DeclareGraphicsExtensions{.eps}
\fi
\usepackage{epsfig,epsf,epstopdf,graphicx}
\usepackage[cmex10]{amsmath}
\usepackage{amssymb}
\usepackage{amsthm }
\usepackage{nicefrac}
\usepackage{hyperref}
\usepackage{xcolor}
\hypersetup{colorlinks=true}
\usepackage{tabularx}

\theoremstyle{theorem}

\theoremstyle{definition}

\theoremstyle{plain}
\newtheorem{Proposition}{Proposition}
\theoremstyle{plain}

\newcommand{\zerovec}{{\bf{0}}}

\newcommand{\xb}{{\textbf{x}}}
\newcommand{\yb}{{\textbf{y}}}

\newcommand{\Bb}{{\textbf{B}}}
\newcommand{\Cb}{{\textbf{C}}}
\newcommand{\sbb}{{\textbf{s}}}

\newcommand{\Xb}{{\textbf{X}}}

\newcommand{\Db}{{\textbf{D}}}
\newcommand{\ub}{{\textbf{u}}}

\newcommand{\eb}{{\textbf{e}}}

\newcommand{\Ab}{{\textbf{A}}}

\newcommand{\Ib}{{\textbf{I}}}

\newcommand{\tb}{{\textbf{t}}}

\newcommand{\Qb}{{\textbf{Q}}}
\newcommand{\Pb}{{\textbf{P}}}

\newcommand{\Gb}{\mathbf{G}}
\newcommand{\Hb}{\mathbf{H}}
\newcommand{\Lb}{\mathbf{L}}

\newcommand{\Ub}{\mathbf{U}}
\newcommand{\Wb}{\mathbf{W}}

\usepackage{algorithm}
\usepackage{multirow}
\usepackage{algcompatible}
\usepackage[bottom]{footmisc}

\PassOptionsToPackage{dvipsnames}{xcolor}
\usepackage{tikz}
\usetikzlibrary{calc}

\ifCLASSOPTIONcompsoc
  \usepackage[caption=false,font=normalsize,labelfont=sf,textfont=sf]{subfig}
\else
  \usepackage[caption=false,font=footnotesize]{subfig}
\fi
\usepackage{fixltx2e}
\usepackage{afterpage}
\usepackage{float}
\usepackage{stfloats}

\makeatletter

\begin{document}

%
\title{Proportionate Adaptive Graph Signal Recovery}

\author{Razieh~Torkamani, Hadi~Zayyani, and Mehdi Korki
\thanks{This work was supported by the Iran National Science Foundation (INSF) (grant number 4005022).}
\thanks{R.~Torkamani and H. Zayyani are with the Department of Electrical and Computer Engineering, Qom University of Technology (QUT), Qom, Iran (e-mail: hut.torkamani@gmail.com, zayyani@qut.ac.ir).}
\thanks{M.~Korki is with School of Science, Computing and Engineering Technologies, Swinburne University of Technology, Melbourne, Australia (e-mail: mkorki@swin.edu.au).}
\vspace{-0.5cm}
}


\maketitle
\thispagestyle{plain}
\pagestyle{plain}

\begin{abstract}
This paper generalizes the proportionate-type adaptive algorithm to the graph signal processing and proposes two proportionate-type adaptive graph signal recovery algorithms. The gain matrix of the proportionate algorithm leads to faster convergence than least mean squares (LMS) algorithm. In this paper, the gain matrix is obtained in a closed-form by minimizing the gradient of the mean-square deviation (GMSD). The first algorithm is the Proportionate-type Graph LMS (Pt-GLMS) algorithm which simply uses a gain matrix in the recursion process of the LMS algorithm and accelerates the convergence of the Pt-GLMS algorithm compared to the LMS algorithm. The second algorithm is the Proportionate-type Graph Extended LMS (Pt-GELMS) algorithm, which uses the previous signal vectors alongside the signal of the current iteration. The Pt-GELMS algorithm utilizes two gain matrices to control the effect of the signal of the previous iterations. The stability analyses of the algorithms are also provided. Simulation results demonstrate the efficacy of the two proposed proportionate-type LMS algorithms.
\end{abstract}

\begin{IEEEkeywords}
Graph signal recovery, Adaptive, Laplacian matrix, least mean-squares, Proportionate.
\end{IEEEkeywords}

%
\IEEEpeerreviewmaketitle


\section{Introduction}
\label{sec:Intro}
Graph Signal Processing (GSP) \cite{Sand13}-\cite{GSP18} is a new research paradigm in signal processing. In GSP, the signal is defined over a graph with irregular domains. It can better represent the inherent structure of the signals defined over nodes. The potential application of GSP includes wireless sensor networks, biological network, social networks, financial networks, and vehicular networks \cite{GSP18}, to name a few. Because of the shift of signal processing to GSP, the available tools in classical signal processing such as shifting, sampling, Fourier transform, filters, etc. is generalized to the graph domain. We can refer to graph sampling, graph signal recovery, graph topology learning, and graph spectral representation as problems within GSP framework. Graph Signal Recovery (GSR) is a basic problem in GSP which aims to recover the whole graph signal by observing signal over only a subset of graph nodes \cite{Chen15}. GSR uses the inherent relationship between the signal values defined over connecting nodes in the graph. This relationship is defined by matrices called weighted adjacency and Laplacian matrix of the graph which will be introduced in sequel.

There are two types of GSR algorithms. The first type is the non-adaptive GSR \cite{Chen15}-\cite{Tork22} which often uses some optimization problems to solve GSR problems in a non-adaptive manner and in a batch-based framework. There is usually a great deal of complexity involved in these algorithms, especially in the large scale graph networks. The second type of GSR algorithms are adaptive GSR algorithms \cite{Loren16}-\cite{Loren18}, which similar to adaptive filter counterparts in classical signal processing, require lower computational complexity and have the potential to perform well in the time-varying nature of the graph signal and noise. Hence, in this paper, we focus on the adaptive GSR algorithms.

Adaptive GSR algorithms first appeared in the literature in 2016. As a pioneering work, a graph Least Mean Square (LMS) algorithm has been developed in \cite{Loren16}, which generalizes the well-known adaptive LMS algorithm to the graph domain. In sequel, \cite{Loren17} suggested a distributed LMS algorithm for learning the graph signal over a network. Moreover, an LMS and a Recursive Least Square (RLS) algorithm have been developed to recover the graph signals from randomly time-varying subset of nodes \cite{Loren18}. Also, in \cite{Loren18EUSIPCO}, the LMS algorithm has been developed for the dynamic graphs in which there is a small perturbation in the Laplacian matrix. In addition, a joint graph weighted adjacency matrix learning and graph signal recovery has been suggested using a Kalman filter for auto-regressive graph signals \cite{Rame18}. To provide scalability and privacy in networks, an online kernel-based graph-adaptive learning algorithm has been suggested in \cite{Shen19}. Besides, \cite{Loren19} proposed a distributed adaptive learning of graph signals using in-network subspace projections. Moreover, a Normalized LMS (NLMS) graph signal estimation algorithm has been proposed in \cite{Spel20}, which has faster convergence than LMS algorithm and has less computational complexity than RLS algorithm. Also, \cite{Ahmadi20} proposed two adaptive GSR algorithms which are Extended LMS (ELMS) algorithm and Fast ELMS (FELMS) algorithm, in which the signal vectors of previous iterations are reused alongside the signal available at the current iteration. In addition, a Single-Kernel Gradraker (SKG) algorithm has been suggested for GSR which uses a Gaussian kernel and it specifies how to find a suitable variance for the kernel \cite{Zhao22}. For adaptive estimation of the graph filter of a graph signal, a graph kernel RLS algorithm has been developed for a different but related problem \cite{Gogi21}.

In this paper, we focus on adaptive graph LMS algorithms for GSR. We generalize the proportionate adaptive filter concept used in adaptive filtering application \cite{Dutt00}-\cite{Wag11}, and used in distributed estimation framework \cite{Yim15}-\cite{ZayyJ21} to the graph domain. Hence, we propose adaptive proportionate GSR algorithm in which a gain matrix is used in the update of the adaptive algorithm. When we have a sparse representation of the graph signal in the domain of Graph Fourier Transforms (GFT), the proposed algorithm can be faster to converge. For determining the gain matrix, the GMSD criterion is used and a closed-form formula is obtained for the gain matrix. The stability analyses of the Graph Proportionate LMS (GPLMS) algorithm are also provided. Finally, simulation results corroborate the efficacy of the proposed algorithm in comparison to some state-of-the-art algorithms in the literature.

This paper is organized as follows. Section~\ref{sec:ProblemForm} presents the essential background on GSP and problem formulation. Section~\ref{sec: PtLMS} presents the proposed proportionate-type LMS algorithm for GSR. The extended proportionate-type LMS algorithm is presented in section~\ref{sec: PtELMS}. Section~\ref{sec: theory} provides some theoretical aspects of the proposed algorithms. In Section~\ref{sec: simulation}, the simulation results are discussed. We conclude the paper in in Section~\ref{sec: conclusion}.

We denote vectors by boldface lowercase letters and matrices by boldface uppercase letters. The operators $(\cdot)^T$ and $(\cdot)^{H}$ represent the transpose and Hermition transpose, respectively. $E[\cdot]$ represents expectation. $\text{diag}( \cdot)$ represents a diagonal matrix with its arguments. The $m$th element of the vector $\mathbf{a}$ is written as $a_m$. The $(m, q)$th element of the matrix $\mathbf{A}$ is written as $a_{mq}$. The identity matrix of dimension $N$ is written as $\mathbf{I}_N$ and the vector $\mathbf{0}_N$ is a length $N$ vector of all zeros. $\mathbf{0}_{M \times N}$ is an $M \times N$ matrix with all entries being zero. The spectral radius of the square matrix $\mathbf{A}$ is denoted by $\rho(\mathbf{A}) \triangleq \max \{|\lambda(\mathbf{A})|\}$.

\section{Problem Formulation and Background}
\label{sec:ProblemForm}

Consider an undirected, connected, weighted graph $G=(\cal V,\cal E)$ with $N$ vertices indexed by $\cal V$$=\{v_1,v_2,...,v_N\}$ and connected together according to the set of edges $\cal E$. The weighted adjacency matrix $\Wb \in \mathbb{R} ^{N\times N}$ is the collection of all edge weights such that $W_{ij}>0$ if $(i,j)\in \cal E$, and $W_{ij}=0$ otherwise. Let $d_i$ be the degree of node $i$ which is defined as $d_i=\sum_{j=1}^N {W_{ij}}$. The degree matrix $\Db$ is a diagonal matrix with the node degrees as its diagonals. Hence, graph Laplacian matrix is defined as $\Lb=\Db-\Wb$.

For undirected graphs, the graph Laplacian is a symmetric and positive semi-deﬁnite matrix. The eigendecomposition for Laplacian matrix is
\begin{equation}
\Lb=\Ub\Lambda\Ub^H,
\end{equation}
where $\Ub\in \mathbb{R} ^{N\times N}$ is the matrix containing all the eigenvectors of $\Lb$ as its columns, and $\Lambda\in \mathbb{R} ^{N\times N}$ is the diagonal matrix of eigenvalues of $\Lb$.

Analogous to the classical signal processing, the graph Fourier transform (GFT) of a graph signal $\xb\in \mathbb{R} ^N$ is defined as its projection onto an orthogonal set of vectors $\{\ub_i\}_{i=1,...,N}$, i.e.
\begin{equation}
\sbb=\Ub^H \xb.
\end{equation}

The basis vector $\{\ub_i\}$ are usually assumed to be the eigenvector set of Laplacian matrix \cite{Loren16},\cite{Ahmadi20}, or the adjacency matrix \cite{Shum13}-\cite{Pesenson08}. Thus, the frequncy-domain representation of the graph signal conveys the intrinsic information of the graph topology. In this paper, we follow the definition of GFT basis vector based on the Laplacian matrix, but the result can be extended to the approach based on the adjacency matrix. The inverse graph Fourier transform (IGFT) can be defined for reconstruction of the graph signal $\xb$ from its frequency domain $\sbb$ as
\begin{equation}
\xb=\Ub \sbb.
\end{equation}

In this paper, we utilize the intrinsic sparsity of bandlimited signals and we assume that the frequency domain representation of the graph signal, i.e. $\sbb$, is sparse. Thus, we can use the compressive sensing (CS) theory and assume the under-sampling of the $\sbb$ as observations. Moreover, we assume that the observation are sampled noisy versions of the graph signal. Thus, using the notations $\xb^o$ and $\sbb^o$ for the original graph signal and its Fourier transform, respectively, we can write the observed signal at time $n$ as
\begin{equation}
\yb[n]=\Bb[n]\Db[n]\xb^o+\eb[n]=\Bb[n]\Db[n]\Ub\sbb^o+\eb[n],
\end{equation}
where $\Bb$ is the CS sensing matrix, and $\Db$ is the sampling matrix. The above equation can be rewritten as
\begin{equation}
\label{eq:observ}
\yb[n]=\Ab[n]\sbb^o+\eb[n],
\end{equation}
where $\Ab[n]=\Bb[n]\Db[n]\Ub$.

\section{Proportionate-type Graph LMS algorithm}
\label{sec: PtLMS}
The LMS algorithm is one of the most popular algorithms in adaptive filtering, and is employed in graph signal recovery \cite{Loren16}-\cite{Loren18EUSIPCO}. Based on the graph-LMS approach, the optimal estimation for the original graph signal can be found by the following recursive procedure
\begin{align}
\sbb[n+1]=\sbb[n]+\mu\Ab^T[n](\yb[n]-\Ab[n]\sbb[n]),
\end{align}
where $\mu$ is the step-size parameter. The proportionate-type LMS (Pt-LMS) algorithm \cite{Dutt00} was introduced as an alternative to the conventional LMS algorithm, which is proved to converge faster than the LMS algorithm by assigning a different gain to each coefficient. This gain is proportional to the magnitude of the coefficient at the current iteration. In this paper, we propose to use this Pt-LMS algorithm in graph signal recovery problem, which results in proportionate-type graph LMS (Pt-GLMS) algorithm. The update equation for the proposed Pt-GLMS algorithm is
\begin{equation}
\label{eq:Pt-LMS}
\sbb[n+1]=\sbb[n]+\mu \Gb[n]\Ab^T[n](\yb[n]-\Ab[n]\sbb[n]),
\end{equation}
where $\Gb[n]$ is the gain matrix. This gain matrix, distinguishes the Pt-LMS from conventional LMS algorithm, and is diagonal, i.e., $\Gb[n]=\mathrm{diag}(g_1[n],...,g_N[n])$, where $g_i[n]$ is the gain factor of the signal of the $i$'th node at time $n$, and is proportional to the magnitude of the graph signal at node $i$, i.e., $g_i[n]\propto |s_i[n]|$. If $g_i[n]=1, \forall i=1,...,N$, i.e., $G[n]=\Ib_{N\times N}$, the Pt-GLMS and GLMS algorithms are equivalent. In the literature \cite{Mula18}, the weights $g_i[n]$ are calculated as
\begin{equation}
\label{eq:gain0}
g_i[n]=\frac{\gamma_i[n]}{\frac{1}{N}\sum_{j=1}^N{\gamma_j[N]}},
\end{equation}
where
\begin{align}
&\gamma_i[n]=\mathrm{max}\{\rho\gamma_{min}[n], F[|s_i[n]|]\},\\
&\gamma_{min}[n]=\mathrm{max}\{\delta, F[|s_1[n]|],...,F[|s_N[n]|]\},
\end{align}
where $\delta$ is the initialization parameter, and the parameter $\rho$ prevents the inactive coefﬁcients from stalling. For $F[|s_i[n]|]$, different functions have been used in the literature \cite{Mula18}.

From (\ref{eq:Pt-LMS}), the factor $\mu g_i[n]$ implies the effective step-size at node $i$, and indicates that in the case of large magnitude of current graph signal, i.e., $|s_i[n]|$, the value of $g_i[n]$ is also large, and, thus, the effective step-size $\mu g_i[n]$ is large, which speed up the convergence of large coefficients. Conversely, for small magnitude of current coefficients, the effective step-size is also small. The resulting proportionate-type graph LMS algorithm is summarized in Algorithm \ref{Algorithm_1}.

\begin{algorithm}[tb]
\caption{Proposed Proportionate-type Graph LMS (Pt-GLMS) Algorithm}
\textbf{Input}   \textbf{Observations} $\yb$; \textbf{Fourier basis functions} $\Ub$; \textbf{Sensing matrix} $\Bb$; \textbf{Sampling matrix} $\Db$; $\mu$; $\rho$; $\delta$; \textbf{number of time instances} $T$ . \newline
\textbf{Initialize} $\sbb=0$, $\Gb=0$, $n=1$.
\label{Algorithm_1}
\begin{algorithmic}
\REPEAT
\begin{itemize}
\item $\Ab[n]=\Bb[n]\Db[n]\Ub$
\item Update $\{g_i[n]\}_{n=1}^N$ using \eqref{eq:gain0} or \eqref{eq:gain11}
\item $\Gb[n]=\mathrm{diag}(g_1[n],...,g_N[n])$
\item $\sbb[n+1]=\sbb[n]+\mu \Gb[n]\Ab^T[n](\yb[n]-\Ab[n]\sbb[n])$
\item $n\longleftarrow n+1$
\end{itemize}
\UNTIL $n\leq T$
\end{algorithmic}
\end{algorithm}
\vspace{-2ex}

\subsection{Gain Matrix Calculation for Pt-GLMS}
\label{sec:gain_1}
In this subsection, we calculate an optimal gain matrix in order to further speed up the convergence of the proposed Pt-GLMS algorithm. To this end, we compute the gradient mean-square deviation (GMSD), and find the optimal gains by minimizing the GMSD at time $n$. The GMSD is defined as \cite{Ahmadi20}
\begin{equation}
\Delta[n]=\mathrm{E}||\tilde{\sbb}[n+1]||^2_Q-\mathrm{E}||\tilde{\sbb}[n]||^2_Q,
\end{equation}
where $\tilde{\sbb}[n]=\sbb^o-\sbb[n]$ is the error signal at time $n$, $Q=\Ab^H\Ab$, and $||\tb||^2_P=\tb^H\Pb\tb$. The GMSD of the $i'$th node at time $n$ can be written as
\begin{equation}
\Delta_i[n]=\mathrm{E}\Big[||\tilde{\sbb}_i[n+1]||^2_Q\Big]-\mathrm{E}\Big[||\tilde{\sbb}_i[n]||^2_Q\Big].
\end{equation}

The optimum gain for node $i$ at time instant $n$ (i.e., $g_i[n]$), by setting the derivative of $\Delta_i[n]$ with respect to $g_i[n]$ equal to zero, is written as
\begin{equation}
\label{eq:gain11}
g_i[n]=
\frac{\mu\Big[\hat{\eb}^T[n]\Ab_i[n]\Ab^T_i[n]\hat{\eb}[n]-\Ab_i[n]\Cb_e\Ab^T_i[n]\Big]}{m^2_i[n]\sum_{j=1}^N a^2_{ji}[n]},
\end{equation}
where
\begin{equation}
\hat{\eb}[n]=\yb[n]-\Ab[n]\sbb[n],
\end{equation}
and
\begin{equation}
\label{eq:gain_15}
m_i[n]=\mu \Ab_i^T[n](\yb[n]-\Ab[n]\sbb[n]).
\end{equation}

\begin{proof}
See Appendix \ref{sec:proof_gain1}.
\end{proof}

Computing the values of $g_i[n]$ using (\ref{eq:gain11})-(\ref{eq:gain_15}), and substituting $\Gb[n]=\mathrm{diag}(g_1[n],...,g_N[n])$, the graph signal can be recovered by (\ref{eq:Pt-LMS}), which yields the minimum GMSD.

As an example, consider a graph with $N=50$ nodes. The spectral content of the graph signal is limited to the first $15$ eigenvectors of the graph Laplacian matrix. The observation noise is drawn from a zero-mean Gaussian distribution with a diagonal covariance matrix. Fig. \ref{fig1} shows the normalized MSD (NMSD), which is calculated as follows
\begin{equation}
\label{eq:NMSD}
\mathrm{NMSD}[n]=\frac{||\sbb^o-\sbb[n]||^2}{||\sbb^o||^2}.
\end{equation}

As can be seen, using proportionate-type algorithm for the reconstruction of graph signals increases the convergence rate of the algorithm. Moreover, the GMSD of the Pt-GLMS using the proposed gain matrix is smaller than that of the others.

\begin{figure}[!t]
\vspace{-5.1ex}
\centering
\hspace*{-1em} {\includegraphics[width=2.3in]{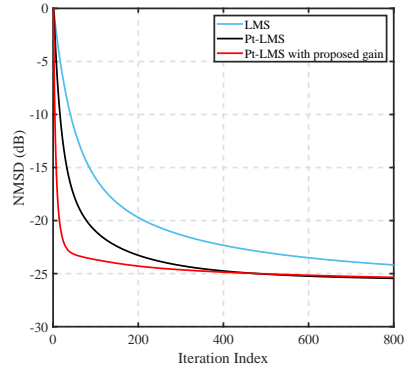} }
\caption{Transient NMSD versus iteration index when estimating a synthetic data with $N=50$ nodes.}
\label{fig1}
\end{figure}

\section{Proportionate-type Graph Extended LMS Algorithm}
\label{sec: PtELMS}
In this subsection, we propose an algorithm to enhance the estimation performance of the Pt-GLMS algorithm proposed in previous section. The proposed algorithm is based on the extended LMS (ELMS) algorithm proposed in \cite{Ahmadi20} and speeds up the convergence of the Pt-GLMS algorithm by using the observations of previous times as well as the current observation. In this work, we propose the proportionate-type of the ELMS algorithm in GSR framework, which assigns the gain matrix to the graph signal estimation process. The resulting algorithm is proportionate-type graph ELMS (Pt-GELMS) algorithm. The update equation for the proposed Pt-GELMS algorithm is
\begin{align}
\label{eq:Pt_ELMS}
\sbb[n+1]=&\sbb[n]+\mu\Gb[n]\Ab^T[n](\yb[n]-\Ab[n]\sbb[n])\nonumber\\
&+\mu\Hb[n]\sum_{j=1}^{K-1}\Ab^T[n-j]\Big(\yb[n-j]-\Ab[n-j]\sbb[n]\Big),
\end{align}
where $K$ is the number of previous time instants used for estimating the graph signal at the current iteration, $\Gb[n]$ and $\Hb[n]$ are diagonal gain matrices assigned to the current and previous observations, respectively. Similar to the previous sections, we have $\Gb[n]=\mathrm{diag}(g_1[n],...,g_N[n])$ and $\Hb[n]=\mathrm{diag}(h_1[n],...,h_N[n])$. Obviously, when $\Hb[n]=\zerovec_{N\times N}$, the Pt-GELMS algorithm is equivalent to the Pt-GLMS algorithm, and when $\Gb[n]=\Hb[n]=\zerovec_{N\times N}$, the Pt-GELMS algorithm reduces to the GLMS algorithm. The resulting proportionate-type graph extended LMS algorithm is summarized in Algorithm \ref{Algorithm_2}.

\subsection{Gain Matrix Calculation for Pt-GELMS}
In this subsection, similar to the process in section \ref{sec:gain_1}, we find an optimal value for the gain matrices by computing and minimizing the GMSD at time $n$. To calculate the GMSD of the $i'$th node at time $n$, the update equation (\ref{eq:Pt_ELMS}) for the $i'$th node can be written as
\begin{equation}
\label{eq:GELMS1}
s_i[n+1]=s_i[n]+g_i[n]m_1[n]+h_i[n]m_2[n],
\end{equation}
where
\begin{align}
\label{eq:GELMS2}
&m_1[n]=\mu\Ab^T_i[n](\yb[n]-\Ab[n]\sbb[n]),\\
&m_2[n]=\mu\sum_{j=1}^{K-1}\Ab^T_i[n-j](\yb[n-j]-\Ab[n-j]\sbb[n]).
\end{align}

Similar to section \ref{sec:gain_1}, computing the GMSD of the $i'$th node at time $n$, i.e. $\Delta_i[n]$, and setting $\frac{\partial \Delta_i[n]}{\partial g_i[n]}=0$ and $\frac{\partial \Delta_i[n]}{\partial h_i[n]}=0$ yields
\begin{equation}
\label{eq:gain22}
g_i[n]=\frac{r_3[n]}{r_4[n]},
\end{equation}
and
\begin{equation}
\label{eq:gain33}
h_i[n]=\frac{r_5[n]}{r_6[n]},
\end{equation}
where
\begin{align}
r_3[n]=&-h_i[n]m_1[n]m_2[n]\sum_{j=1}^{K-1}{a^2_{ji}[n]}-\mu\Ab_i[n]\Cb_e\Ab_i^T[n]\nonumber\\
&+\mu (\yb[n]-\Ab[n]\sbb[n])^T\Ab_i[n]\Ab_i^T[n](\yb[n]-\Ab[n]\sbb[n]),\\
r_4[n]=&m_1^2[n]\sum_{j=1}^{K-1}{a^2_{ji}[n]},\\
r_5[n]=&-g_i[n]m_1[n]m_2[n]\sum_{j=1}^{K-1}{a^2_{ji}[n]}\nonumber\\
&-\mu\Ab_i[n]\Cb_e\sum_{j=1}^{K-1}{\Ab_i^T[n-j]}+\mu(\yb[n]-\Ab[n]\sbb[n])^T\nonumber\\
&\times\Ab_i[n]\Ab_i^T[n]\sum_{j=1}^{K-1}{(\yb[n-j]-\Ab[n-j]\sbb[n])},\\
r_6[n]=&m_2^2[n]\sum_{j=1}^{K-1}{a^2_{ji}[n]}.
\end{align}

\begin{proof}
See Appendix \ref{sec:proof_gain2}.
\end{proof}

Obviously, $g_i[n]$ and $h_i[n]$ are interdependent, and, thus, should be updated in a repetitive manner.

\begin{algorithm}[tb]
\caption{Proposed Proportionate-type Graph extended LMS (Pt-GELMS) Algorithm}
\textbf{Input}   \textbf{Observations} $\yb$; \textbf{Fourier basis functions} $\Ub$; \textbf{Sensing matrix} $\Bb$; \textbf{Sampling matrix} $\Db$; $\mu$; $K$; \textbf{number of time instances} $T$ . \newline
\textbf{Initialize} $\sbb=0$, $\Gb=0$, $\Hb=0$, $n=1$.
\label{Algorithm_2}
\begin{algorithmic}
\REPEAT
\begin{itemize}
\item $\Ab[n]=\Bb[n]\Db[n]\Ub$
\item Update $\{g_i[n]\}_{n=1}^N$ and $\{h_i[n]\}_{n=1}^N$ using \eqref{eq:gain22} and \eqref{eq:gain33}, respectively
\item $\Gb[n]=\mathrm{diag}(g_1[n],...,g_N[n])$ and $\Hb[n]=\mathrm{diag}(h_1[n],...,h_N[n])$
\item $\sbb[n+1]=\sbb[n]+\mu\Gb[n]\Ab^T[n](\yb[n]-\Ab[n]\sbb[n])$\newline
$+\mu\Hb[n]\sum_{j=1}^{K-1}\Ab^T[n-j]\Big(\yb[n-j]-\Ab[n-j]\sbb[n]\Big)$
\item $n\longleftarrow n+1$
\end{itemize}
\UNTIL $n\leq T$
\end{algorithmic}
\end{algorithm}
\vspace{-2ex}

\section{Theoretical Analysis}
\label{sec: theory}
In this section, some theoretical analysis of the proposed algorithms are provided. Since the Pt-GLMS is a special case of Pt-GELMS, we only discuss the theoretical analysis of Pt-GELMS, and the analysis for the Pt-GLMS can be achieved by setting $\Hb[n]=\zerovec_{N\times N}$.

\subsection{Mean-Square Analysis}
In this subsection, we study the mean-square behavior of the proposed algorithms. To compute MSD, note that using the definition $\tilde{\sbb}[n]=\sbb^o-\sbb[n]$, and subtracting both sides of (\ref{eq:Pt_ELMS}) from $\sbb^o$, we obtain
\begin{equation}
\label{eq:tsbb}
\tilde{\sbb}[n+1]=\Big[\Ib-\mu \Bb_1[n]\Big]\tilde{\sbb}[n]+\mu \Bb_2[n]+\mu\Bb_3[n],
\end{equation}
where
\begin{align}
&\Bb_1[n]=\Gb[n] \Ab^T[n] \Ab[n]+\Hb[n]\sum_{j=1}^{K-1}{\Ab^T[n-j]\Ab[n-j]},\\
&\Bb_2[n]=\Gb[n]\Ab^T[n]\eb[n],\\
&\Bb_3[n]=\Hb[n] \sum_{j=1}^{K-1}{\Ab^T[n-j]\eb[n-j]}.
\end{align}

\begin{Proposition}
\label{prop:1}
Assume that the noise samples at different times are independent.  Then, for any bounded initial condition, the proposed Pt-GELMS algorithm asymptotically converges in the mean-square error sense if
\begin{equation}
0<\mu<\frac{2}{\lambda_{max}\left(\Bb_1\right)},
\end{equation}
where $\lambda_{max}(.)$ denotes the maximum eigenvalue of the matrix therein.
\end{Proposition}

\begin{proof}
See Appendix~\ref{sec:proof_prop1}.
\end{proof}

\subsection{Mean Performance}
In this subsection, we study the steady-state mean performance of the proposed Pt-GELMS algorithm. Taking the expectation of both sides of \eqref{eq:tsbb} yields
\begin{equation}
\mathrm{E}\Big[\tilde{\sbb}[n+1]\Big]=\Big[\Ib-\mu \Bb_1[n]\Big]\mathrm{E}\Big[\tilde{\sbb}[n]\Big].
\end{equation}

As the algorithm reaches the steady-state, we have $\mathrm{E}\Big[\tilde{\sbb}[n+1]\Big]\longrightarrow\mathrm{E}\Big[\tilde{\sbb}[n]\Big]$. Thus, to guarantee the convergence in the mean sense, the condition $\rho(\Ib-\mu \Bb_1[n])<1$ should be satisfied, which implies
\begin{equation}
0<\mu<\frac{2}{\lambda_{max}\left(\Bb_1\right)}.
\end{equation}

\subsection{Steady-State Performance}
Taking the limit of (\ref{eq:MSD1}) as $n\to\infty$ yields
\begin{equation*}
\lim_{n \rightarrow \infty}
\mathrm{E}\Big[||\tilde{\sbb}[n]||^2_{(\Ib-\Qb)\boldsymbol{\phi}}\Big]=\mu^2\mathrm{vec}(\Pb)^T\boldsymbol{\phi}.
\end{equation*}

To evaluate the steady-state MSD of the proposed Pt-GELMS algorithm (see (\ref{eq:Pt_ELMS})), we can set $\phi=(\Ib-\Qb)^{-1}\mathrm{vec}(\Ib)$ and obtain
\begin{displaymath}
\mathrm{MSD}=
\lim_{n \rightarrow \infty} \mathrm{E}||\tilde{\xb}[n]||^2=\lim_{n \rightarrow \infty}\mathrm{E}||\tilde{\sbb}[n]||^2
\end{displaymath}
\begin{equation}
=\mu^2\mathrm{vec}(\Pb)^T(\Ib-\Qb)^{-1}\mathrm{vec}(\Ib).
\end{equation}

\section{Simulation Results}
\label{sec: simulation}

\begin{figure}[!t]
\centering
\hspace*{-1em} {\includegraphics[width=2.3in]{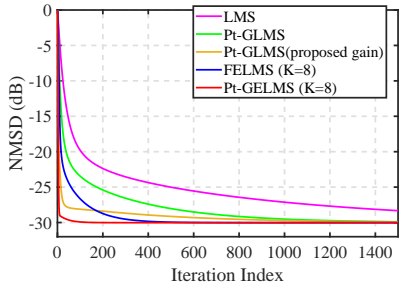} }
\caption{Adaptive algorithms performance versus iteration index for $K=8$, $M=30$, and $|S|=20$.}
\label{fig2}
\end{figure}

\begin{figure}[!t]
\centering
\hspace*{-1em} {\includegraphics[width=2.3in]{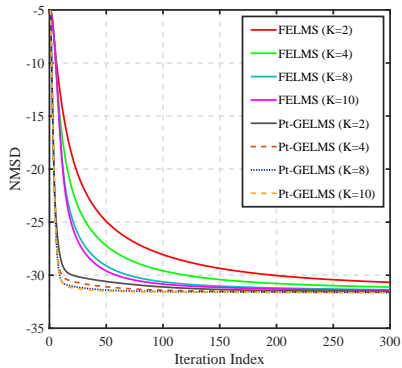} }
\caption{Transient NMSD versus iteration index for $M=30$, $|S|=20$, and different values of $K$.}
\label{fig3}
\end{figure}

In this section, we evaluate the performance of the proposed algorithms via some numerical results on synthetic and real data.

In the first scenario, we consider a synthetic graph with $N=50$ nodes. Similar to \cite{Pu21}, the edge weights are drawn randomly from a uniform distribution $W_{ij}\sim\mathcal{U}(0,1)$, and the weighted adjacency matrix is derived as $\Wb=(\Wb+\Wb^T)/2$. The graph Laplacian matrix is calculated as $\Lb=\Db-\Wb$, where $\Db$ is the diagonal degree matrix with the diagonal elements as $d_i=\sum_{j=1}^N{W_{ij}}$. The spectral content of the graph is limited to the first $15$ eigenvectors of the graph Laplacian matrix. The observation noise in (\ref{eq:observ}) is drawn from a zero-mean Gaussian distribution with a diagonal covariance matrix $\Cb_e=\sigma_e^2\Ib$ with $\sigma_e^2=0.01$. We compare the results of our proposed algorithms with that of the graph LMS algorithm \cite{Loren16}, and extended proportionate-type LMS \cite{Ahmadi20}. Fig. \ref{fig2} shows the transient behavior of the NMSD in (\ref{eq:NMSD}) versus the iteration index for the number of previous time instants $K=8$, number of CS measurements $M=30$, number of samples $|S|=20$, and step-size $\mu=0.01$. Each point in the curves is the result of ensemble average over 50 independent simulations. It can be seen that the proposed proportionate-type LMS algorithms converge faster than the conventional LMS. Moreover, using the proposed algorithm for estimation of gain matrices lead to a further increase in the convergence rate.

\begin{figure*}[!t]
\centering
\subfloat[$M=10$]{\includegraphics[width=1.65in]{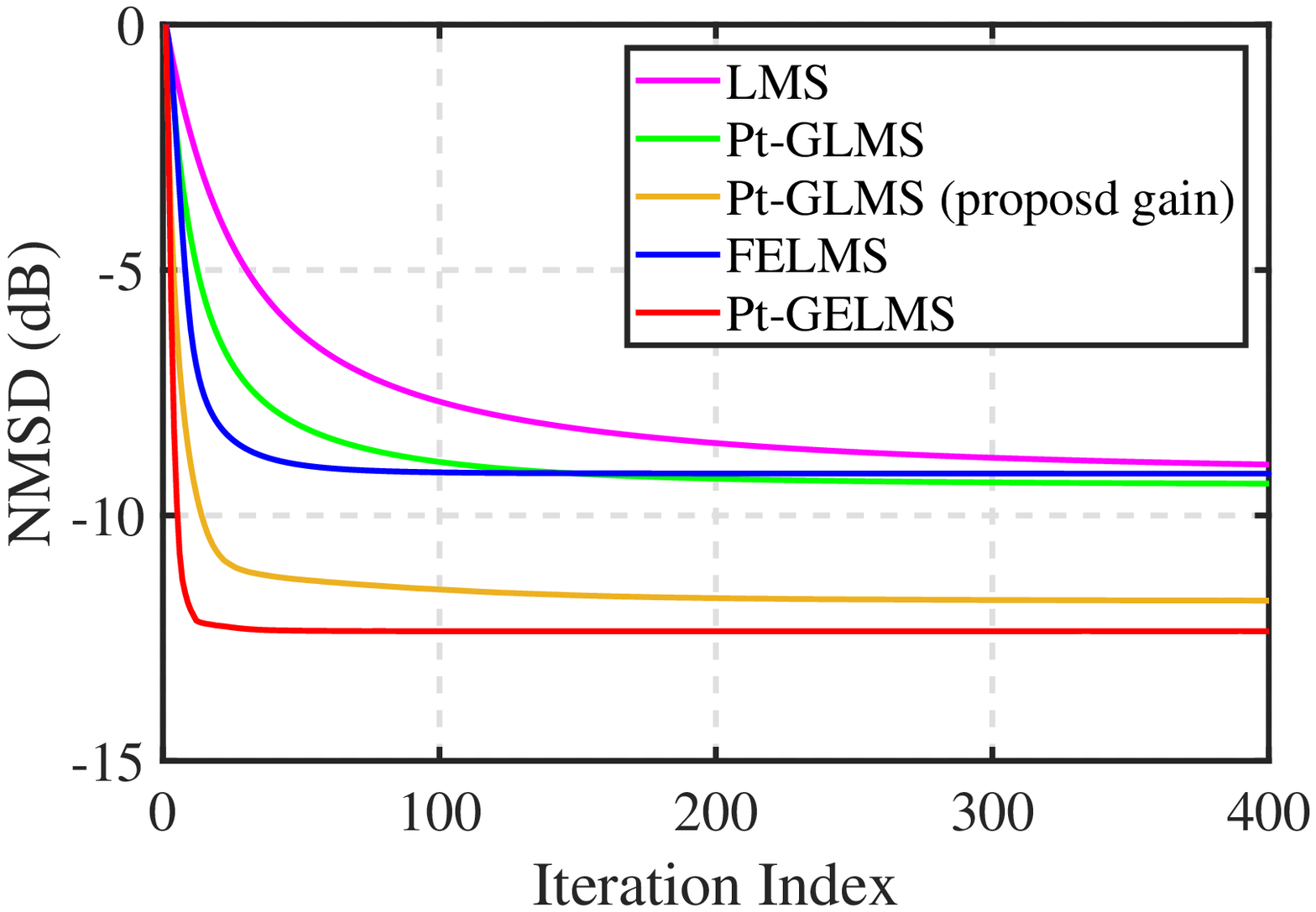}\hspace{-0.4cm}%
\label{fig4:a}}
\subfloat[$M=30$]{\includegraphics[width=1.65in]{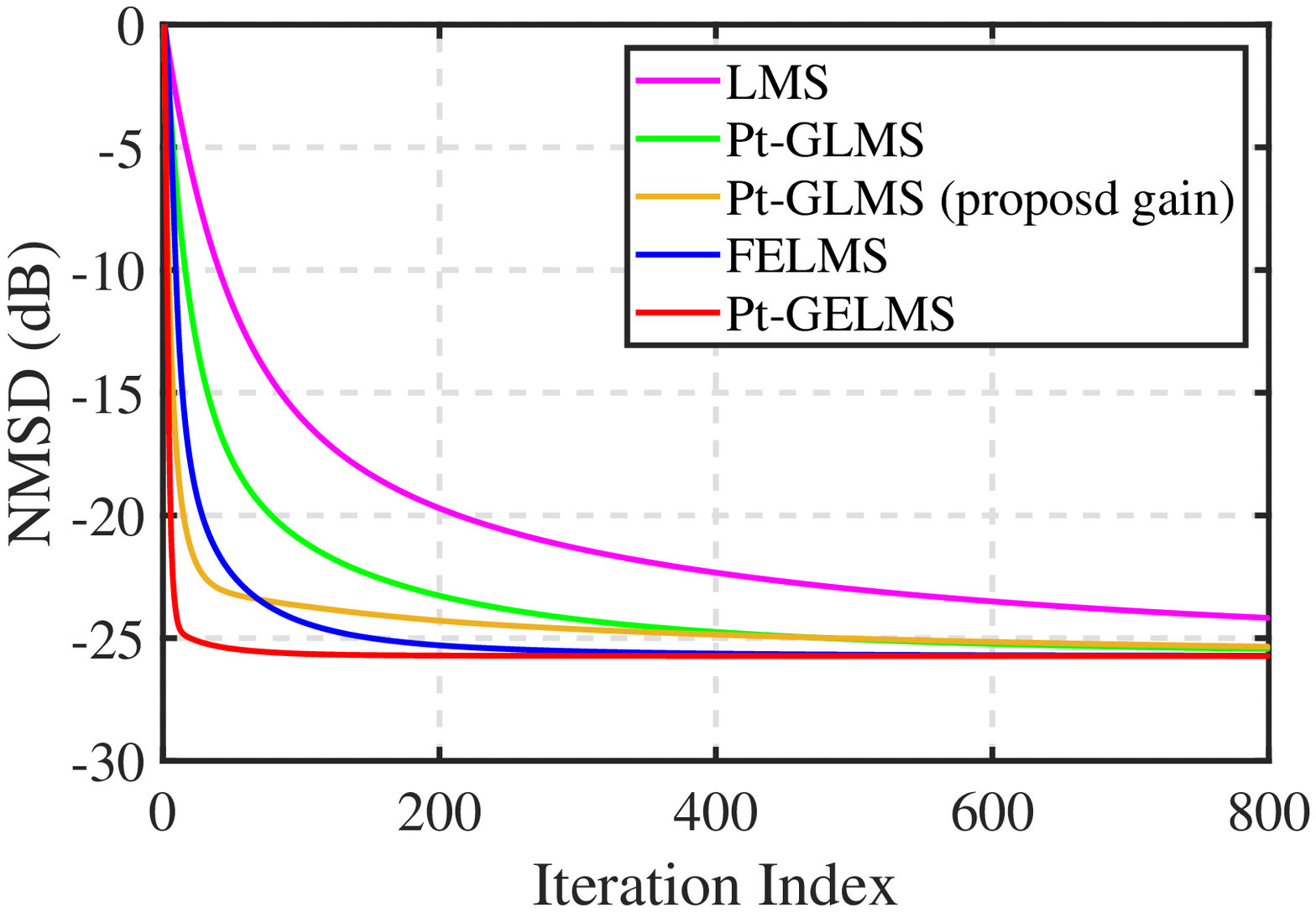}\hspace{-0.4cm}%
\label{fig4:b}}
\subfloat[$M=40$]{\includegraphics[width=1.65in]{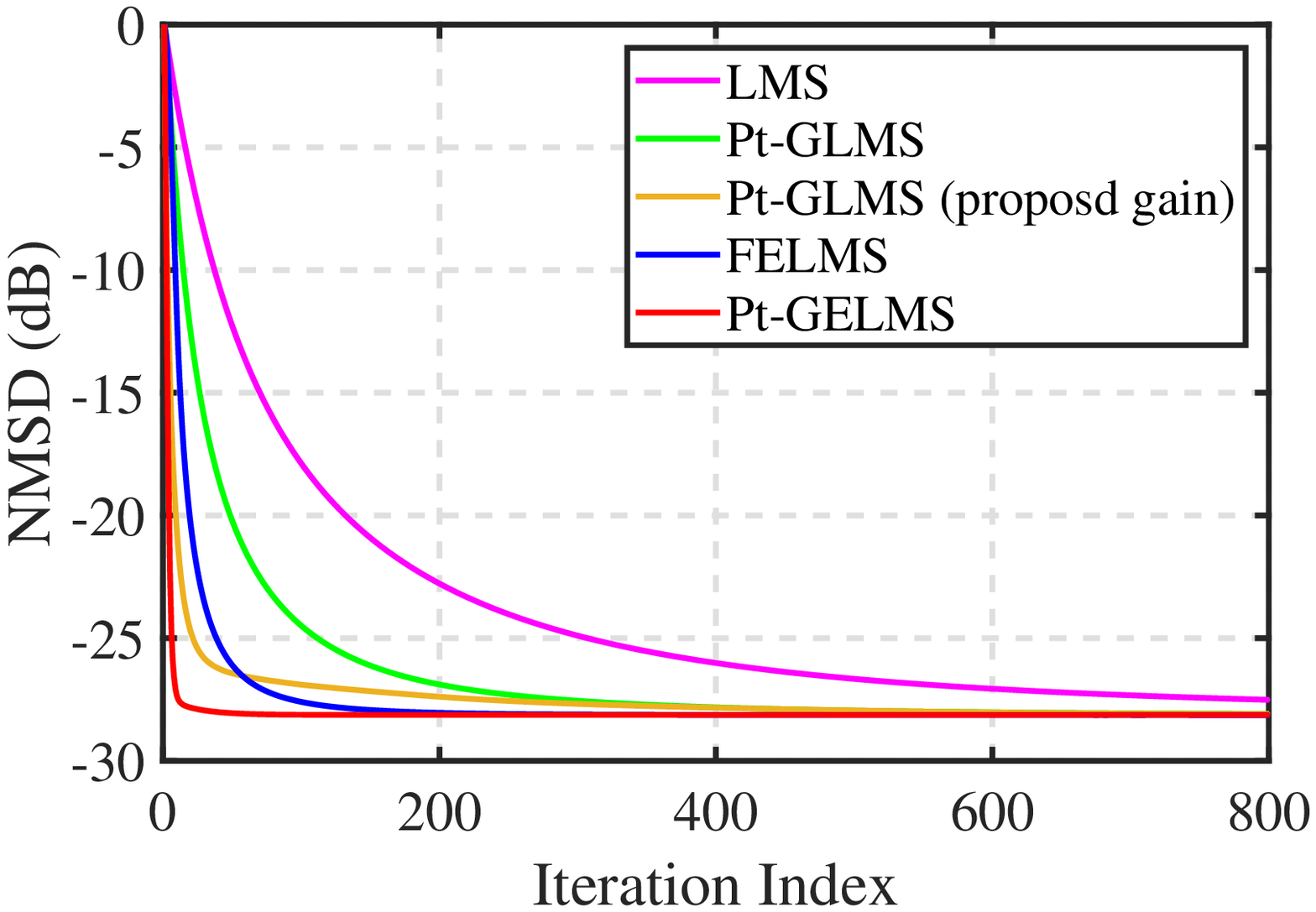}%
\label{fig4:c}}
\caption{Transient NMSD behaviour for $K=8$, $|S|=20$, and different values of $M$.}
\label{fig4}
\end{figure*}

In Fig. \ref{fig3}, the transient behaviour of the NMSD versus the iteration index are shown for $M=30$, step-size $\mu=0.01$, and for different values of $K$. The curves are averaged over $50$ independent trials. It can be seen that increasing the value of the parameter $K$, which incorporates further previous signals in estimating the current signal, results in higher convergence speed of the algorithms.

Fig. \ref{fig4} depicts the transient behaviour of the NMSD for LMS, FELMS, Pt-GLMS (with the proposed gain matrix), and Pt-GELMS algorithms, considering different number of CS measurements, i.e., different values of $M$, and $K=8$, $|S|=20$, and the step-size $\mu=0.01$. The results are averaged over 50 independent trials. As can bee seen from the curves, increasing the number of CS measurements, which means increasing the number of linear combination of signal ensembles in hand, leads to an increase in convergence rate and decrease in NMSD. Moreover, with the same number of measurements, the proposed Pt-GLMS and Pt-GELMS algorithms outperform their peer algorithms, i.e., LMS and FELMS, respectively.

In Fig. \ref{fig5}, the transient behaviour of the NMSD for LMS, FELMS, Pt-GLMS (with the proposed gain matrix), and Pt-GELMS algorithms are presented, considering three different values of bandwidth, $|F|$, and $K=8$, $|S|=20$, $M=30$, and the step-size $\mu=0.01$. The results are averaged over 50 independent trials. As can bee seen from the curves, increasing the bandwidth of the signal spectrum, leads to an increase in convergence rate and decrease in NMSD.

\begin{figure*}[!t]
\centering
\subfloat[$|F|=10$]{\includegraphics[width=1.65in]{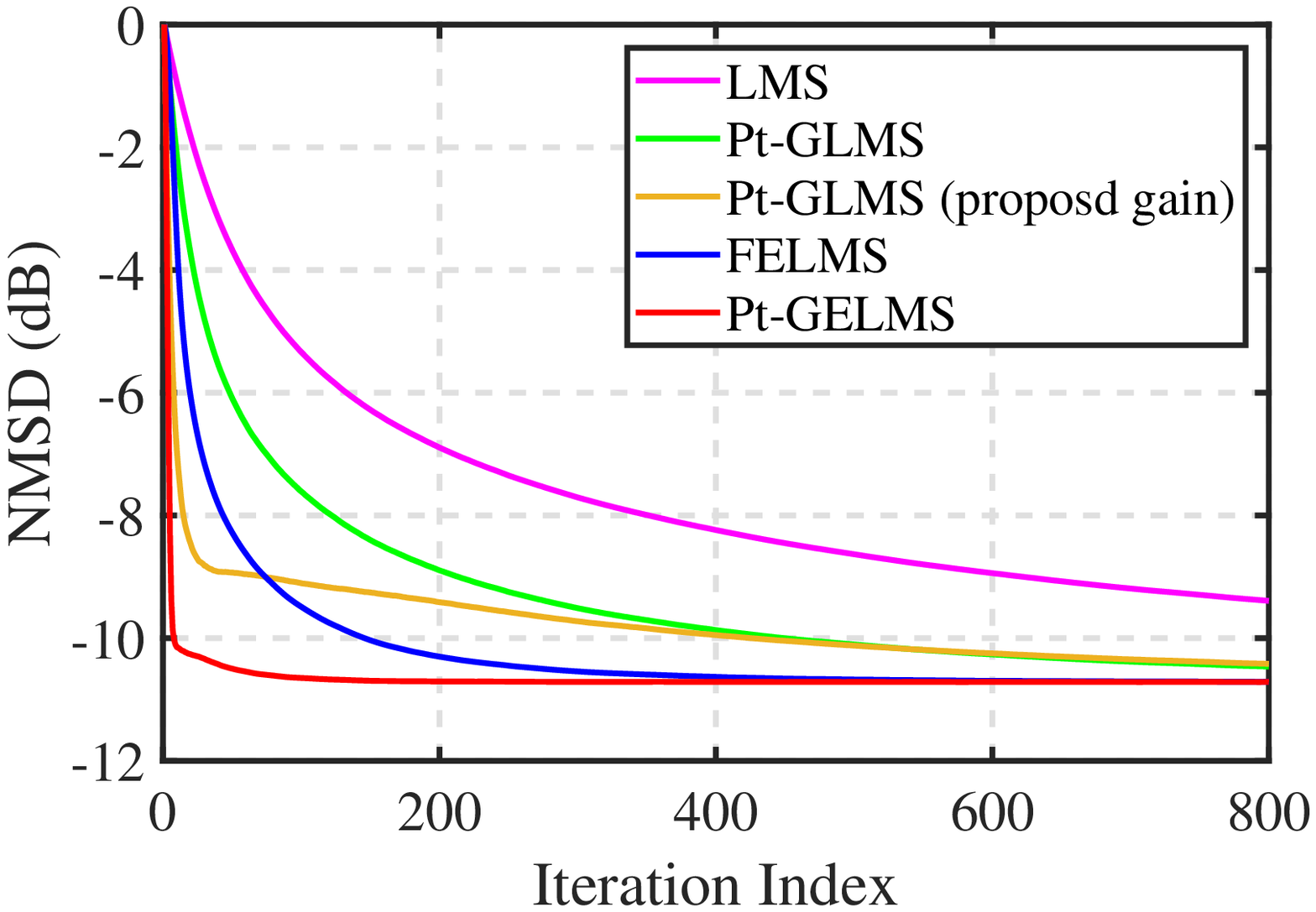}%
\label{fig5:a}}
\subfloat[$|F|=15$]{\includegraphics[width=1.65in]{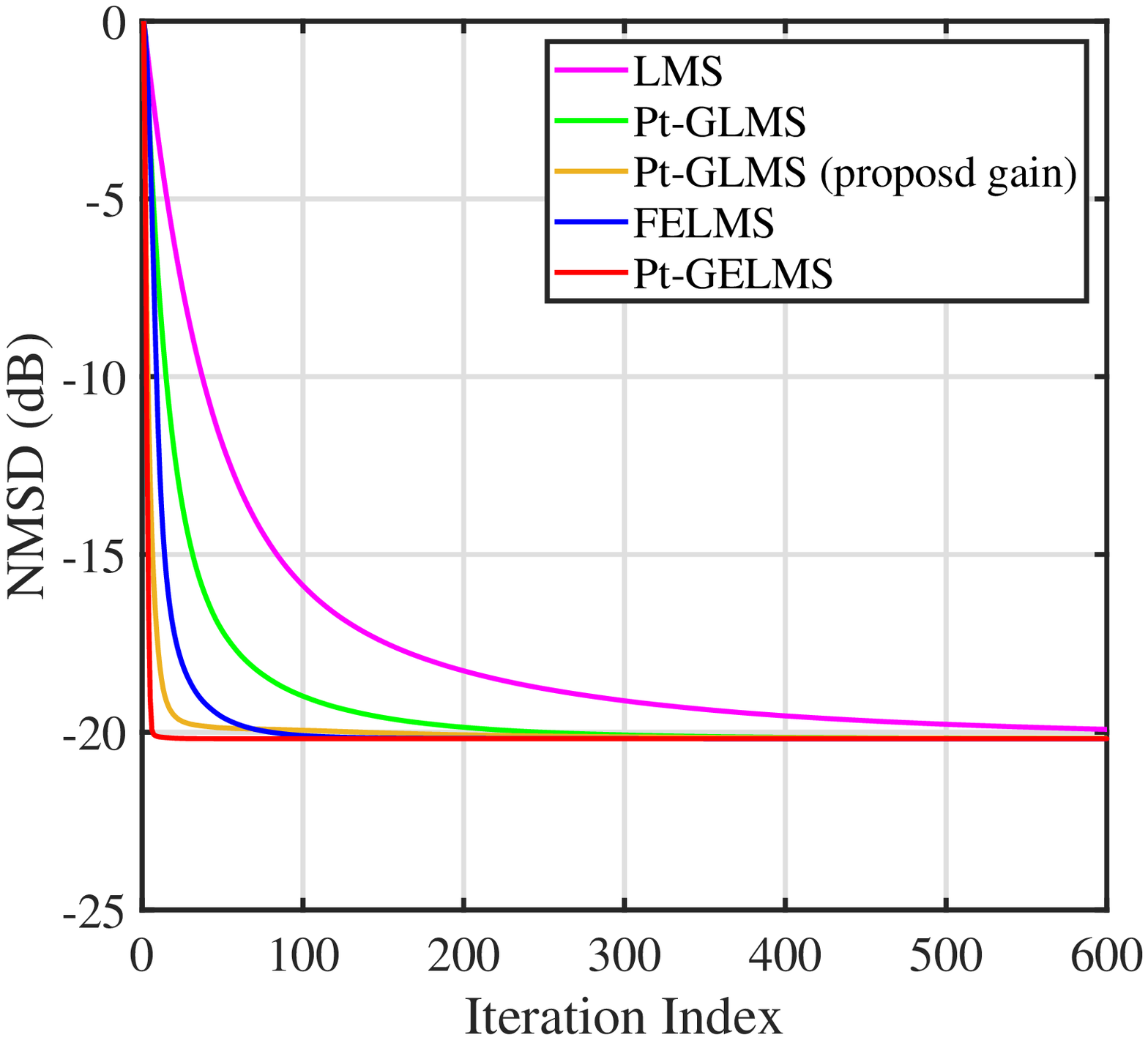}%
\label{fig5:b}}
\subfloat[$|F|=20$]{\includegraphics[width=1.65in]{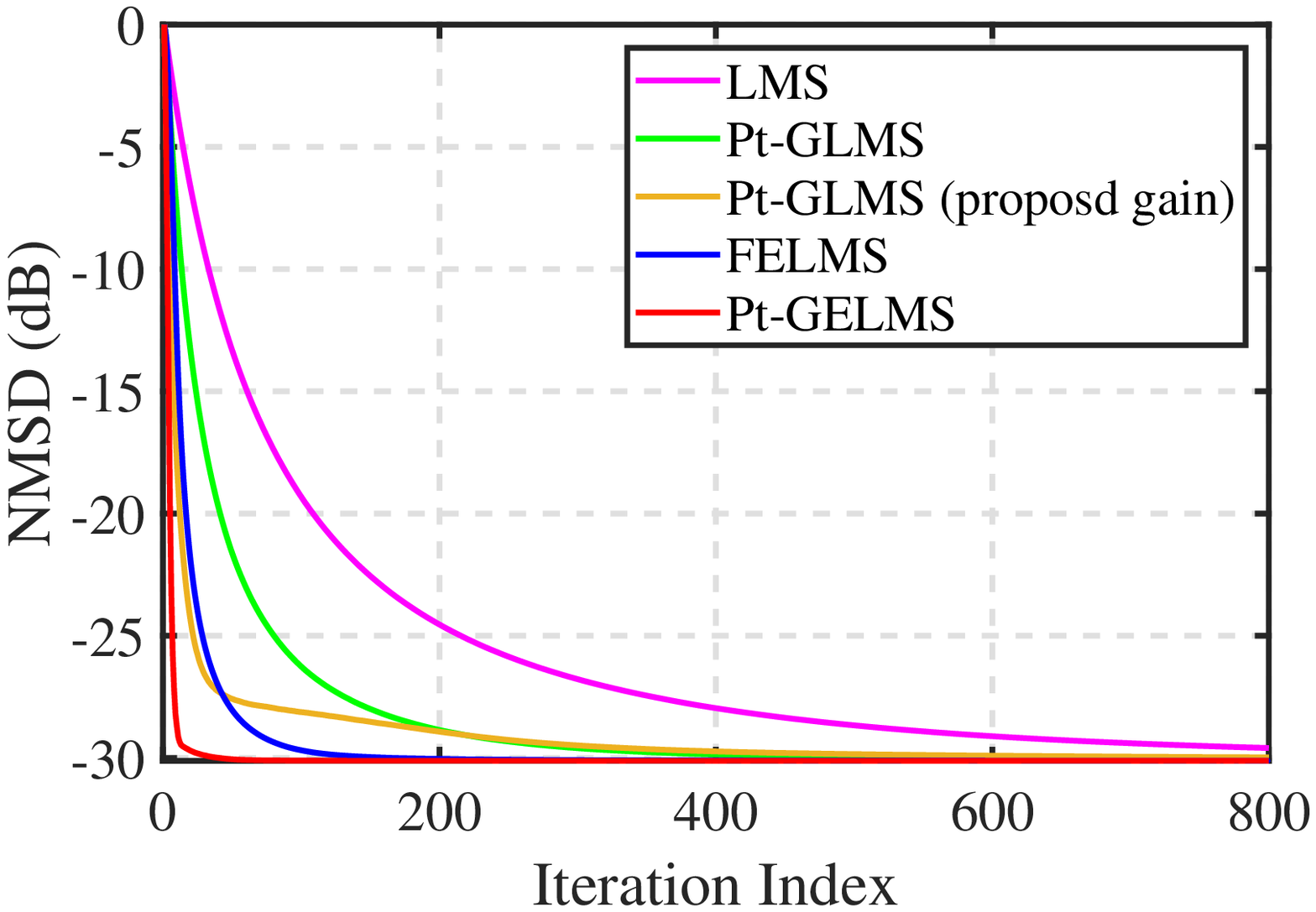}%
\label{fig5:c}}
\caption{Transient NMSD of the adaptive algorithms for $K=8$, $|S|=20$, $M=30$, and different values of $|F|$.}
\label{fig5}
\end{figure*}

\begin{figure*}[!t]
\centering
\subfloat[$|S|=10$]{\includegraphics[width=1.65in]{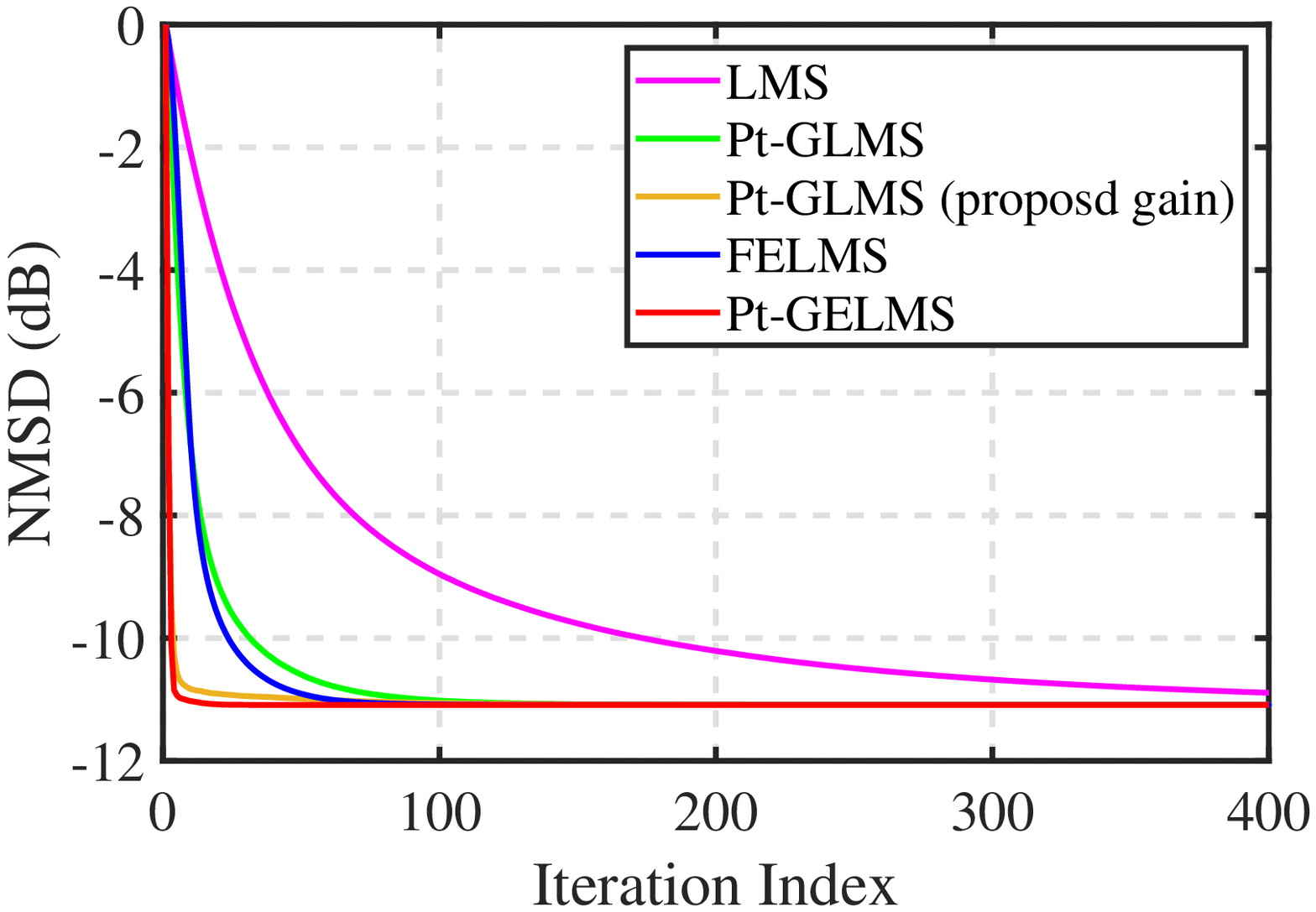}%
\label{fig6:a}}
\subfloat[$|S|=20$]{\includegraphics[width=1.65in]{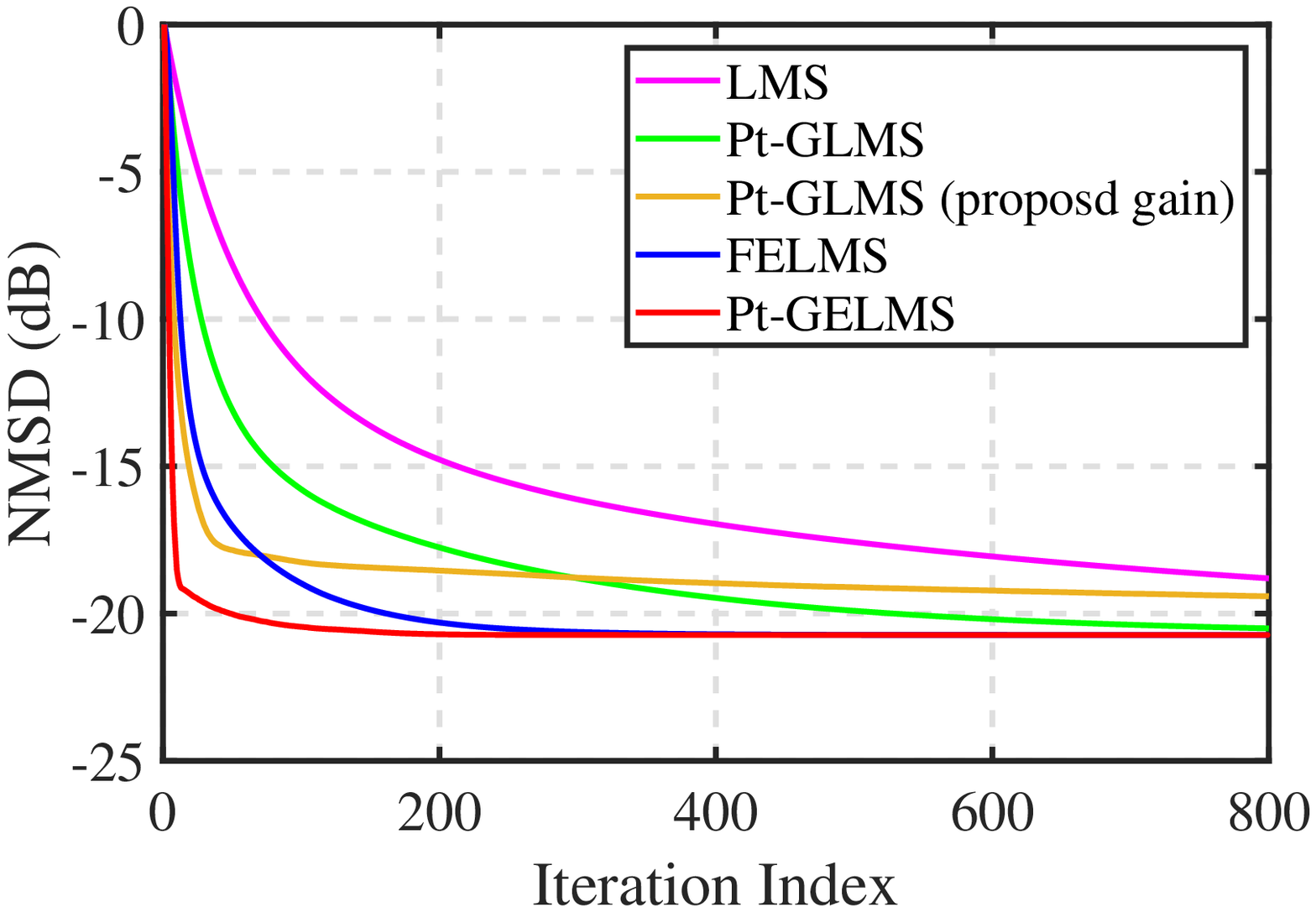}%
\label{fig6:b}}
\subfloat[$|S|=30$]{\includegraphics[width=1.65in]{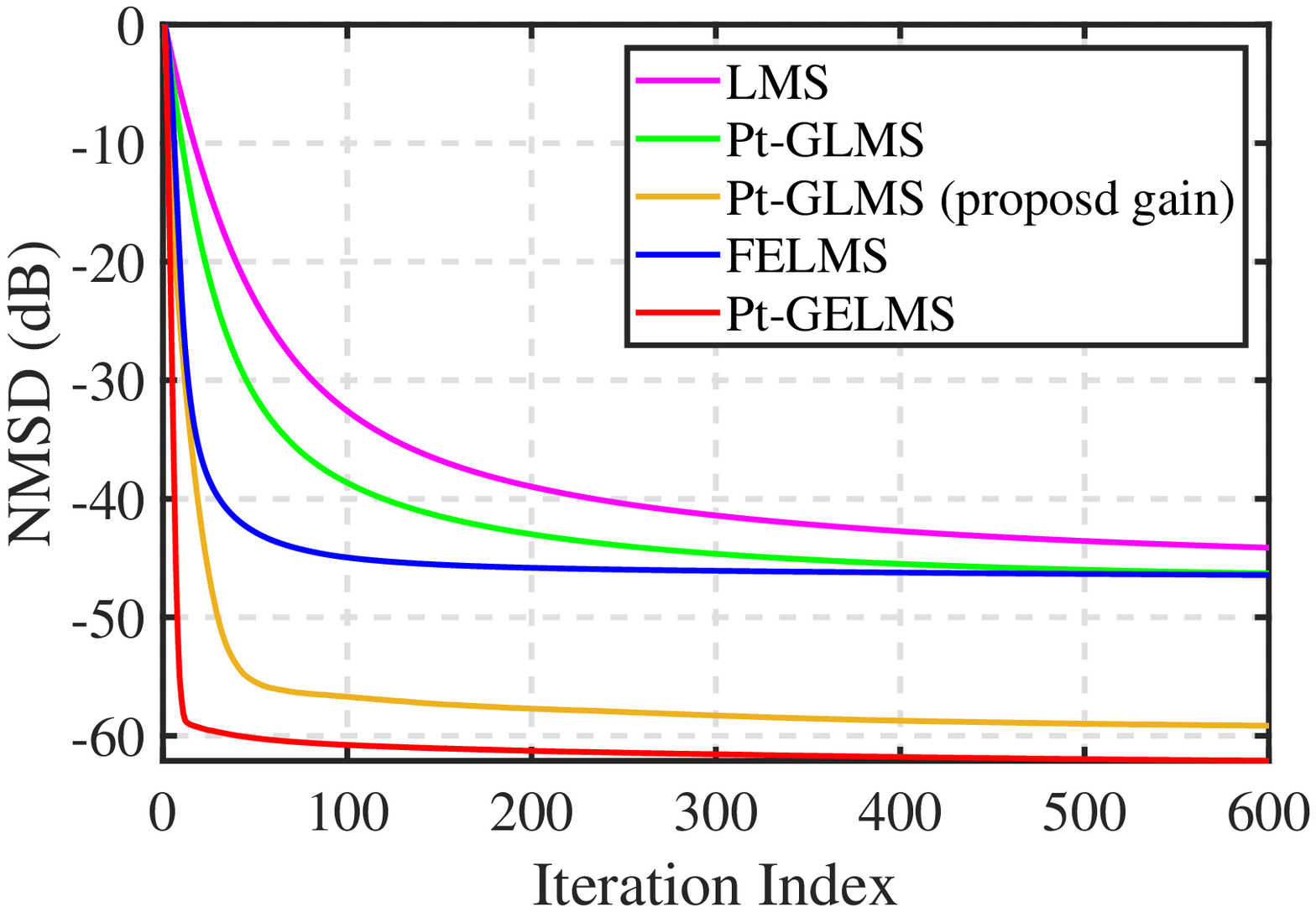}%
\label{fig6:c}}
\caption{Transient NMSD of the adaptive algorithms versus iteration index for $K=8$, $|F|=15$, $M=30$, and different values of $|S|$.}
\label{fig6}
\end{figure*}

In Fig. \ref{fig6}, the transient behaviour of the NMSD for LMS, FELMS, Pt-GLMS (with the proposed gain matrix), and Pt-GELMS algorithms are presented, considering three different number of samples selected, $|S|$, and $K=8$, $|F|=15$, $M=30$, and the step-size $\mu=0.01$. The results are averaged over 50 independent trials. As can bee seen from the curves, increasing number of nodes selected in the sampling procedure, leads to an increase in convergence rate and decrease in NMSD. As seen, both the proposed Pt-GELMS and Pt-GLMS algorithms outperforms their non-proportionate type peers, i.e., FELMS and LMS. Furthermore, in all cases, the Pt-GLMS with the proposed gain matrix perform better than the Pt-GLMS with conventional gain matrix used in the literature.

In the second scenario, for assessing the performance of the adaptive algorithms with real-world data, we consider the temperature graph signal. In this case, we use the dataset downloaded from the Intel Berkeley Research lab (refer to \cite{Intel04}, \cite{Tork21}, \cite{Tork22}), in which the temperature values from a total of $N=54$ sensors are acquired. We aim to estimate the temperature values using the competing adaptive methods. Similar to the first case, we add an observation noise which is a zero-mean Gaussian noise with a diagonal covariance matrix $\Cb_e=\sigma_e^2\Ib$ with $\sigma_e^2=3$ (which yields $SNR\sim 25  dB$). The transient NMSD versus the iteration index is depicted in Fig. \ref{fig7} for three different numbers of selected nodes ($M=30, 40, 54$), where we have used $K=6$. The figures show the superiority of the proposed algorithm in comparison to the other adaptive algorithms.

\begin{figure*}[!t]
\centering
\subfloat[$M=30$]{\includegraphics[width=1.65in]{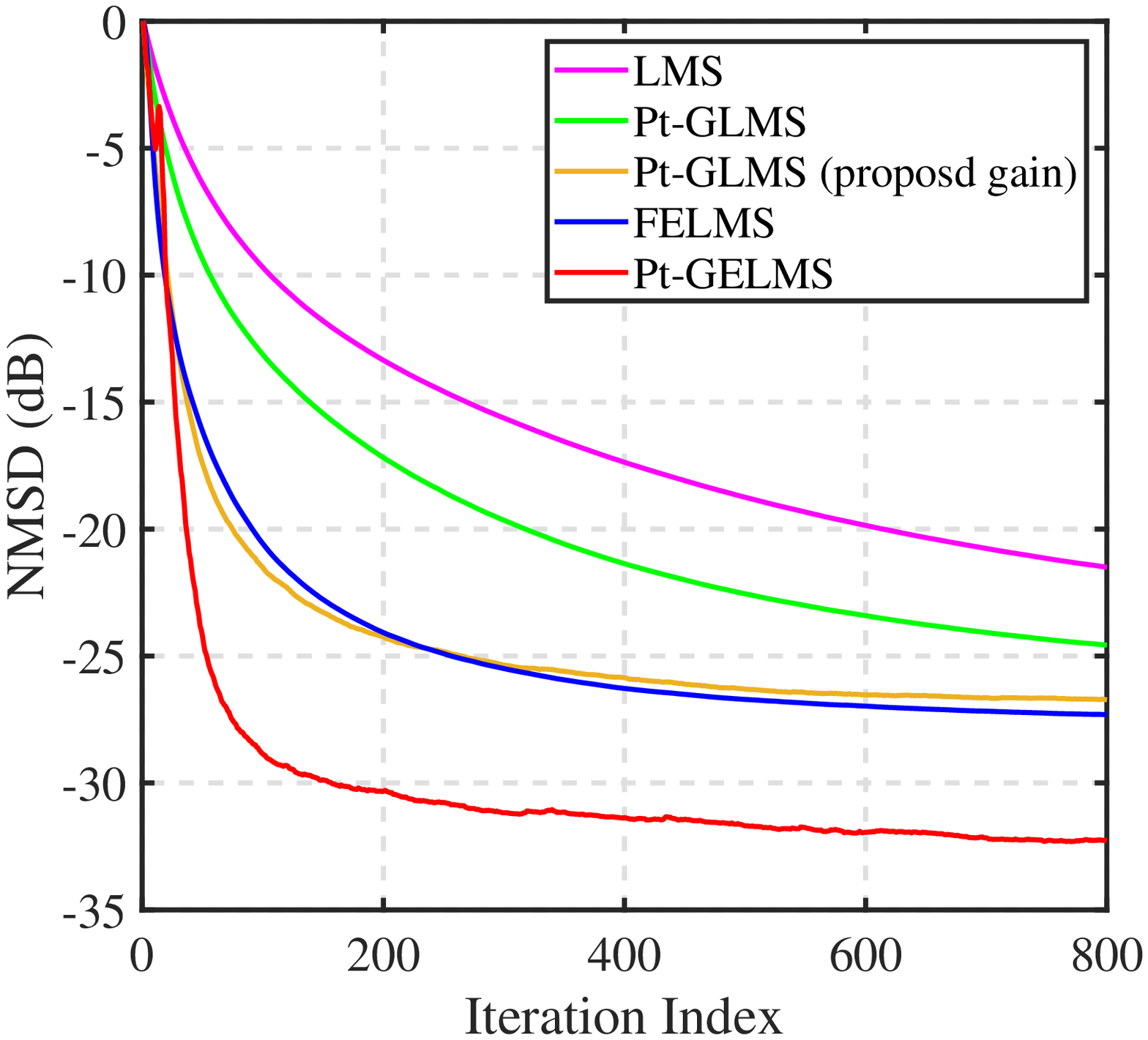}%
\label{fig6:a}}
\subfloat[$M=40$]{\includegraphics[width=1.65in]{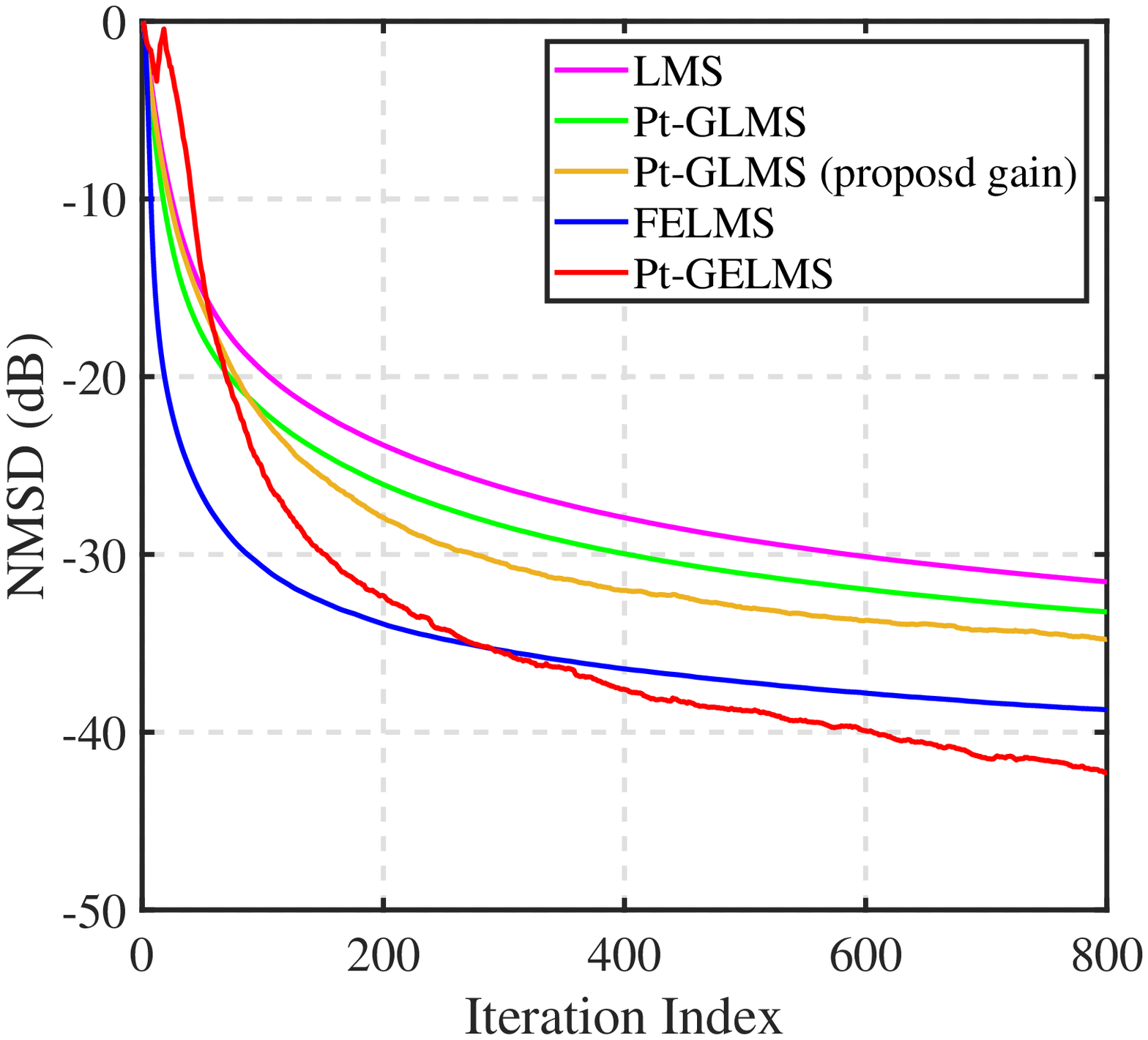}%
\label{fig6:b}}
\subfloat[$M=N(=54)$]{\includegraphics[width=1.65in]{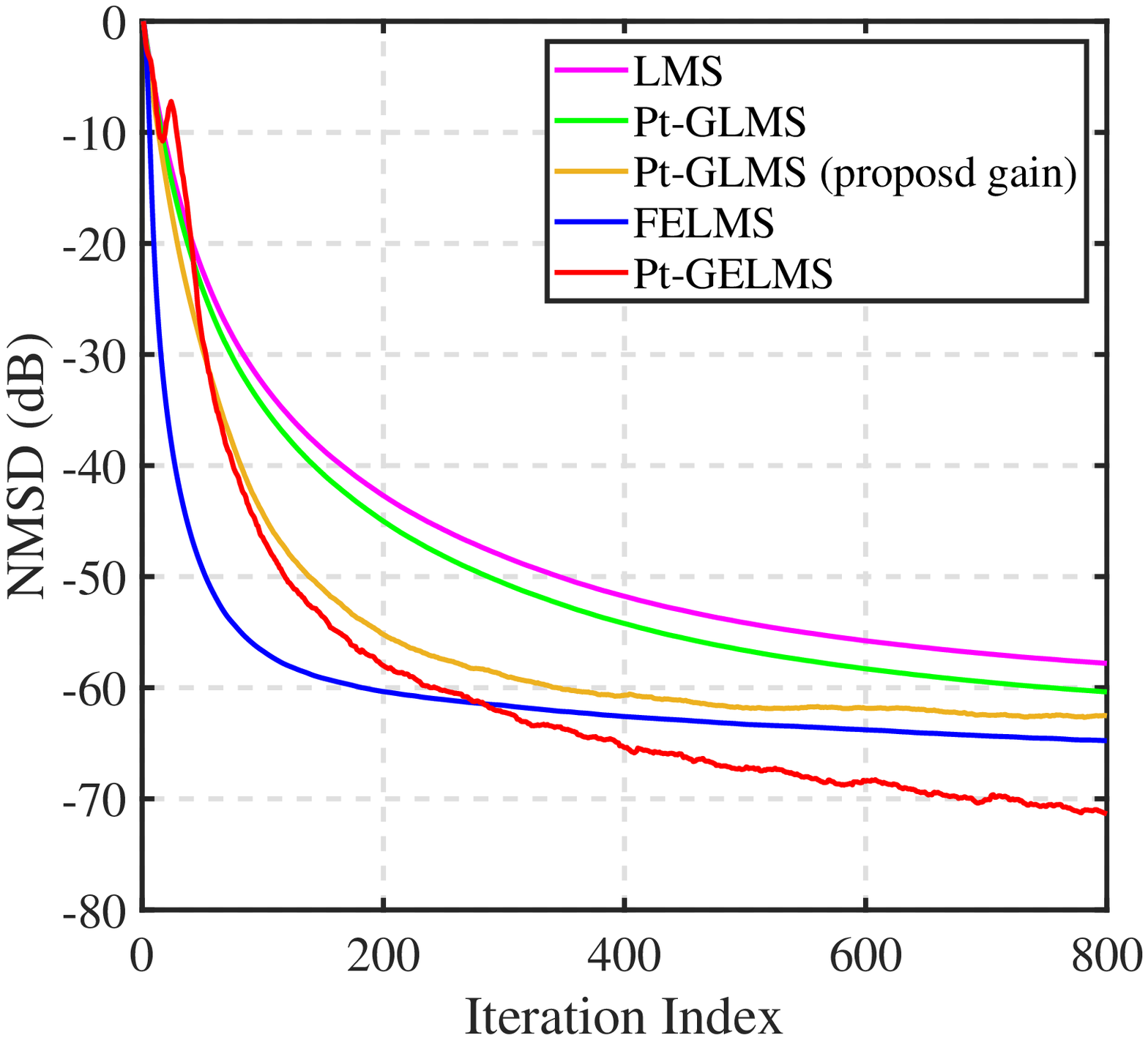}%
\label{fig6:c}}
\caption{Adaptive algorithms performance when estimating the temperature data from \cite{Intel04}.}
\label{fig7}
\end{figure*}

\section{Conclusion}
\label{sec: conclusion}
In this paper, we formulated the proportionate adaptive graph signal recovery algorithm. Hence, the proportionate adaptive filtering algorithm in the classical signal processing is generalized to proportionate adaptive GSR algorithm in GSP. Two proportionate-type algorithms are proposed for adaptive GSR, of which the first is the Pt-GLMS algorithm and the second is the Pt-GELMS algorithm. In two algorithms, the gain matrix (or matrices) are obtained optimally in a closed form via minimizing the GMSD. Some theoretical analysis of the proposed algorithms such as mean-square convergence analysis, mean performance analysis, and steady-state performance analysis are also presented. Simulation results in both synthetic and real data, demonstrate the faster convergence of the proposed proportionate-type algorithms in comparison to non-proportionate-type counterparts.

\begin{appendices}
\section{Proof of Gain Matrix for Pt-GLMS}
\label{sec:proof_gain1}
The GMSD of the $i'$th node at time $n$ can be written as
\begin{equation}
\Delta_i[n]=\mathrm{E}\Big[||\tilde{\sbb}_i[n+1]||^2_Q\Big]-\mathrm{E}\Big[||\tilde{\sbb}_i[n]||^2_Q\Big].
\end{equation}
Moreover, from (\ref{eq:Pt-LMS}), we have
\begin{equation}
s_i[n+1]=s_i[n]+\mu g_i[n]\Ab_i^T[n](\yb[n]-\Ab[n]\sbb[n]),
\end{equation}
where $\Ab_i$ is the $i'$th column of matrix $\Ab$. We rewrite the above equation in terms of the signal errors as
\begin{equation}
\tilde{s}_i[n+1]=\tilde{s}_i[n]-g_i[n]m_i[n],
\end{equation}
where
\begin{equation}
m_i[n]=\mu \Ab_i^T[n](\yb[n]-\Ab[n]\sbb[n]).
\end{equation}

Thus, we have
\begin{displaymath}
||\tilde{s}_i[n+1]||^2_Q=||\tilde{s}_i[n]-g_i[n]m_i[n]||^2_Q
\end{displaymath}
\begin{eqnarray*}
&=&\Big[(\tilde{s}_i[n]-g_i[n]m_i[n])^2\sum_{j=1}^N a^2_{ji}[n]\Big]\nonumber\\
&&+\: 2\Big[(\tilde{s}_i[n]-g_i[n]m_i[n])\sum_{j=1}^N\sum^N_{\substack{k=1\\k\neq i}}\tilde{s}_k[n]a_{ji}[n]a_{jk}[n]\Big]\nonumber\\
&=&\Big[\tilde{s}^2_i[n]\sum_{j=1}^N a^2_{ji}[n]\Big]+\Big[g^2_i[n]m^2_i[n]\sum_{j=1}^N a^2_{ji}[n]\Big]\nonumber\\
&&-\: 2\Big[\tilde{s}_i[n]g_i[n]m_i[n]\sum_{j=1}^N a^2_{ji}[n]\Big]\nonumber\\
&&+\: 2\Big[\tilde{s}_i[n]\sum_{j=1}^N\sum^N_{\substack{k=1\\k\neq i}}\tilde{s}_k[n] a_{ji}[n]a_{jk}[n]\Big]\nonumber\\
&&-\: 2\Big[g_i[n] m_i[n]\sum_{j=1}^N\sum^N_{\substack{k=1\\k\neq i}}\tilde{s}_k[n] a_{ji}[n]a_{jk}[n]\Big].
\end{eqnarray*}
It can be shown that the sum of the first and fourth terms in the above equation is $||\tilde{s}_i[n]||^2_Q$. Moreover, the sum of the third and fifth terms is $-2g_i[n] m_i[n] \tilde{s}^{T}[n]\Ab^{T}[n]\Ab_i[n]$. Thus, we have
\begin{equation}
\label{delta1}
\Delta_i[n]=f_1[n]-2f_2[n],
\end{equation}
where
\begin{align}
&f_1[n]=\mathrm{E}\Big[g^2_i[n] m^2_i[n]\sum_{j=1}^N a^2_{ji}[n]\Big]\\
&f_2[n]=\mathrm{E}\Big[g_i[n] m_i[n] \tilde{s}^{T}[n]\Ab^{T}[n]\Ab_i[n]\Big].
\end{align}

It can be seen that we can calculate the GMSD if $f_1[n]$ and $f_2[n]$ are given. On the other hand, since the signal error vector $\tilde{\sbb}$ is not available, we can not compute $f_2[n]$ directly. To solve this problem, we rewrite $f_2[n]$ as
\begin{align}
f_2[n]&=\mathrm{E}\Big[g_i[n] m_i[n] (\Ab[n]\tilde{\sbb}[n])^T\Ab_i[n]\Big]\nonumber\\
&=\mathrm{E}\Big[g_i[n] m_i[n] (\Ab[n](\sbb^o-\sbb[n]))^T\Ab_i[n]\Big]\nonumber\\
&=\mathrm{E}\Big[g_i[n] m_i[n](\yb[n]-\eb[n]-\Ab[n]\sbb[n])^T\Ab_i[n]\Big].
\end{align}
Assuming that the noise $\eb[n]$ is independent of the other processes, we have
\begin{align}
f_2[n]=&\mu\mathrm{E}\Big[g_i[n](\yb[n]-\Ab[n]\sbb[n])^T\Ab_i[n]\Ab^T_i[n](\yb[n]-\Ab[n]\sbb[n])\Big]\nonumber\\
&-\mu \mathrm{E}\Big[g_i[n]\Ab_i[n]\Cb_e\Ab^T_i[n]\Big].
\end{align}

Thus, the GMSD becomes as
\begin{align}
&\Delta_i[n]=\mathrm{E}\Big[g^2_i[n] m^2_i[n] \sum_{j=1}^N a^2_{ji}[n]\Big]\nonumber\\
&-2\mu\mathrm{E}\Big[g_i[n](\yb[n]-\Ab[n]\sbb[n])^T\Ab_i[n]\Ab^T_i[n](\yb[n]-\Ab[n]\sbb[n])\Big]\nonumber\\
&+2\mu\mathrm{E}\Big[g_i[n]\Ab_i[n]\Cb_e\Ab^T_i[n]\Big].
\end{align}

The optimum gain for node $i$ at time instance $n$ (i.e., $g_i[n]$) can be found by setting the derivative of $\Delta_i[n]$ with respect to $g_i[n]$ to zero and obtain
\begin{equation}
\label{eq:gain1}
g_i[n]=
\frac{\mu\Big[\hat{\eb}^T[n]\Ab_i[n]\Ab^T_i[n]\hat{\eb}[n]-\Ab_i[n]\Cb_e\Ab^T_i[n]\Big]}{m^2_i[n]\sum_{j=1}^N a^2_{ji}[n]},
\end{equation}
where $\hat{\eb}[n]=\yb[n]-\Ab[n]\sbb[n]$.

\section{Proof of Gain Matrix for Pt-GELMS}
\label{sec:proof_gain2}

From \eqref{eq:GELMS1}, the signal errors are
\begin{equation}
\tilde{s}_i[n+1]=\tilde{s}_i[n]-g_i[n]m_1[n]-h_i[n]m_2[n],
\end{equation}
where $m_1[n]$ and $m_2[n]$ are given by \eqref{eq:GELMS2}. Thus, we have
\begin{align*}
\mathrm{E}||\tilde{s}_i[n+1]&||^2_Q=\mathrm{E}||\tilde{s}_i[n]-g_i[n]m_1[n]-h_i[n]m_2[n]||^2_Q\nonumber\\
=&\mathrm{E}\Big[\tilde{s}^2_i[n]\sum_{j=1}^{K-1} a^2_{ji}[n]\Big]+\mathrm{E}\Big[g^2_i[n] m^2_1[n] \sum_{j=1}^{K-1} a^2_{ji}[n]\Big]\nonumber\\
&+\mathrm{E}\Big[h^2_i[n] m^2_2[n] \sum_{j=1}^{K-1} a^2_{ji}[n]\Big]\nonumber\\
&+2\mathrm{E}\Big[g_i[n]h_i[n]m_1[n]m_2[n]\sum_{j=1}^{K-1} a^2_{ji}[n]\Big]\nonumber\\
&-2\mathrm{E}\Big[\tilde{s}_i[n] g_i[n]m_1[n]\sum_{j=1}^{K-1} a^2_{ji}[n]\Big]\nonumber\\
&-2\mathrm{E}\Big[\tilde{s}_i[n] h_i[n]m_2[n]\sum_{j=1}^{K-1} a^2_{ji}[n]\Big]\nonumber\\
&+2\mathrm{E}\Big[\tilde{s}_i[n]\sum_{j=1}^{K-1}\sum^{K-1}_{\substack{k=1\\k\neq i}}\tilde{s}_k[n] a_{ji}[n] a_{jk}[n]\Big]\nonumber\\
&-2\mathrm{E}\Big[g_i[n]m_1[n]\sum_{j=1}^{K-1}\sum^{K-1}_{\substack{k=1\\k\neq i}}\tilde{s}_k[n] a_{ji}[n] a_{jk}[n]\Big]\nonumber\\
&-2\mathrm{E}\Big[h_i[n]m_2[n]\sum_{j=1}^{K-1}\sum^{K-1}_{\substack{k=1\\k\neq i}}\tilde{s}_k[n] a_{ji}[n] a_{jk}[n]\Big].
\end{align*}

Therefore, the GMSD becomes as
\begin{eqnarray*}
\label{eq:rr}
\Delta_i[n]&=&\mathrm{E}||\tilde{s}_i[n+1]||^2_Q-\mathrm{E}||\tilde{s}_i[n]||^2_Q\nonumber\\
&=&\mathrm{E}\Big[(g^2_i[n]m^2_1[n]+2g_i[n]h_i[n]m_1[n]m_2[n]\nonumber\\
&&+\:h^2_i[n]m^2_2[n])\sum_{j=1}^{K-1} a^2_{ji}[n]\Big]\nonumber\\
&&-\:2\mathrm{E}\Big[\tilde{s}_i[n](g_i[n]m_1[n]+h_i[n]m_2[n])\sum_{j=1}^{K-1} a^2_{ji}[n]\Big]\nonumber\\
&&-\:2\mathrm{E}\Big[(g_i[n]m_1[n]+h_i[n]m_2[n])\nonumber\\
&&\times\:\sum_{j=1}^{K-1}\sum^{K-1}_{\substack{k=1\\k\neq i}} \tilde{s}_k[n] a_{ji}[n]a_{jk}[n]\Big]\nonumber\\
&=&r_1[n]-2r_2[n],
\end{eqnarray*}
where
\begin{align}
r_1[n]=&\mathrm{E}\Big[(g^2_i[n]m^2_1[n]+2g_i[n]h_i[n]m_1[n]m_2[n]\nonumber\\
&+h^2_i[n]m^2_2[n])\sum_{j=1}^{K-1} a^2_{ji}[n]\Big],\nonumber\\
r_2[n]=&\mathrm{E}\Big[(g_i[n]m_1[n]+h_i[n]m_2[n])\tilde{\sbb}^T[n]\Ab^T[n]\Ab_i[n]\Big].\nonumber\\
\end{align}

Similar to section \ref{sec:gain_1}, the signal error vector $\tilde{\sbb}[n]$ is not available, and thus, $r_2[n]$ is not computable. Thus, we solve this problem using the following relations
\begin{equation}
\Ab\tilde{\sbb}[n]=\Ab(\sbb^o-\sbb[n])=\yb[n]-\eb[n]-\Ab[n]\sbb[n],
\end{equation}
which yields
\begin{align}
r_2[n]=&\mathrm{E}\Big[(\Ab[n]\tilde{\sbb}[n])^T\Ab_i[n](g_i[n]m_1[n]+h_i[n]m_2[n])\Big]\nonumber\\
=&\mu\mathrm{E}\Big[g_i[n](\yb[n]-\Ab[n]\sbb[n])^T\Ab_i[n]\Ab^T_i[n](\yb[n]-\Ab[n]\sbb[n])\Big]\nonumber\\
&+\mu\mathrm{E}\Big[h_i[n](\yb[n]-\Ab[n]\sbb[n])^T\Ab_i[n]\sum_{j=1}^{K-1} \Ab^T_i[n-j]\Big]\nonumber\\
&-\mu\mathrm{E}\Big[g_i[n]\Ab_i[n]\Cb_e\Ab^T_i[n]\Big]\nonumber\\
&-\mu\mathrm{E}\Big[h_i[n]\Ab_i[n]\Cb_e\sum_{j=1}^{K-1} \Ab^T_i[n-j]\Big].
\end{align}

Substituting the above equation in (\ref{eq:rr}) and setting $\frac{\partial \Delta_i[n]}{\partial g_i[n]}=0$ and $\frac{\partial \Delta_i[n]}{\partial h_i[n]}=0$ yields
\begin{equation}
\label{eq:gain2}
g_i[n]=\frac{r_3[n]}{r_4[n]},
\end{equation}
and
\begin{equation}
\label{eq:gain3}
h_i[n]=\frac{r_5[n]}{r_6[n]},
\end{equation}
where
\begin{align}
r_3[n]=&-h_i[n]m_1[n]m_2[n]\sum_{j=1}^{K-1}{a^2_{ji}[n]}-\mu\Ab_i[n]\Cb_e\Ab_i^T[n]\nonumber\\
&+\mu (\yb[n]-\Ab[n]\sbb[n])^T\Ab_i[n]\Ab_i^T[n](\yb[n]-\Ab[n]\sbb[n]),\\
r_4[n]=&m_1^2[n]\sum_{j=1}^{K-1}{a^2_{ji}[n]},\\
r_5[n]=&-g_i[n]m_1[n]m_2[n]\sum_{j=1}^{K-1}{a^2_{ji}[n]}\nonumber\\
&-\mu\Ab_i[n]\Cb_e\sum_{j=1}^{K-1}{\Ab_i^T[n-j]}+\mu(\yb[n]-\Ab[n]\sbb[n])^T\nonumber\\
&\times\Ab_i[n]\Ab_i^T[n]\sum_{j=1}^{K-1}{(\yb[n-j]-\Ab[n-j]\sbb[n])},\\
r_6[n]=&m_2^2[n]\sum_{j=1}^{K-1}{a^2_{ji}[n]}.
\end{align}

\section{Proof of Proposition \ref{prop:1}}
\label{sec:proof_prop1}.

Let $\boldsymbol{\Phi}\in \mathbb{R}^{N\times N}$ be an arbitrary matrix. Then, $\boldsymbol{\Phi}-$weighted norm of both sides of (\ref{eq:tsbb}) yields
\begin{align}
\mathrm{E}\Big[||\tilde{\sbb}[n+1]||&^2_{\boldsymbol{\Phi}}\Big]=\mathrm{E}\Big[||\tilde{\sbb}[n]||^2_{\boldsymbol{\Phi}'}\Big]\nonumber\\
&+\mu^2\mathrm{E}\Big[\eb^T[n]\Ab[n]\Gb^T[n]\boldsymbol{\Phi}\Gb[n]\Ab^T[n]\eb^T[n]\Big]\nonumber\\
&+\mu^2\mathrm{E}\Big[\sum_{j=1}^{K-1}{\eb^T[n-j]\Ab[n-j]}\Hb^T[n]\boldsymbol{\Phi}\nonumber\\
&\times\Hb[n]\sum_{j=1}^{K-1}{\Ab^T[n-j]\eb[n-j]}\Big],
\end{align}
where
\begin{equation}
\boldsymbol{\Phi}'=(\Ib-\mu\Bb_1[n])^T\boldsymbol{\Phi}(\Ib-\mu\Bb_1[n]).
\end{equation}
and
\begin{displaymath}
\Bb_1[n]=\Gb[n] \Ab^T[n] \Ab[n]+\Hb[n]\sum_{j=1}^{K-1}{\Ab^T[n-j]\Ab[n-j]},\\
\end{displaymath}

Let $\boldsymbol{\phi}=\mathrm{vec}(\boldsymbol{\Phi})$, where the operator $\mathrm{vec}(.)$ aggregates the columns of the matrix therein on top of each other; then, we have
\begin{align}
\label{eq:msd0}
\mathrm{E}\Big[||\tilde{\sbb}[n+1]&||^2_{\boldsymbol{\phi}}\Big]=\mathrm{E}\Big[||\tilde{\sbb}[n]||^2_{\Qb\boldsymbol{\phi}}\Big]\nonumber\\
&+\mu^2\mathrm{Tr}\Big[\boldsymbol{\Phi}\Gb[n]\Ab^T[n]\Cb_e\Ab[n]\Gb^T[n]\Big]\nonumber\\
&+\mu^2\mathrm{Tr}\Big[\boldsymbol{\Phi} \Hb[n]\sum_{j=1}^{K-1}{\Ab^T[n-j]\Cb_e\Ab[n-j]}\Hb^T[n]\Big],
\end{align}
where $\mathrm{Tr}(.)$ is the trace operator, and we have
\begin{equation}
\Qb=(\Ib-\mu\Bb_1[n])\otimes(\Ib-\mu\Bb_1[n]).
\end{equation}

In (\ref{eq:msd0}), we have used the independency assumption among noise samples at different times. Exploiting the trace property $\mathrm{Tr}(\boldsymbol{\Phi}\Xb)=\mathrm{vec}(\Xb^H)^T\mathrm{vec}(\boldsymbol{\Phi})$ in the above equation, we obtain
\begin{align}
\label{eq:MSD1}
\mathrm{E}\Big[||\tilde{\sbb}[n+1]&||^2_{\boldsymbol{\phi}}\Big]=\mathrm{E}\Big[||\tilde{\sbb}[n]||^2_{\Qb\boldsymbol{\phi}}\Big]\nonumber\\
&+\mu^2\mathrm{vec}\Big[\Gb[n]\Ab^T[n]\Cb_e\Ab[n]\Gb^T[n]\Big]\boldsymbol{\phi}\nonumber\\
&+\mu^2\mathrm{vec}\Big[\Hb[n]\sum_{j=1}^{K-1}{\Ab^T[n-j]\Cb_e\Ab[n-j]}\Hb[n]\Big]\boldsymbol{\phi}\nonumber\\
=&\mathrm{E}\Big[||\tilde{\sbb}[n]||^2_{\Qb\boldsymbol{\phi}}\Big]+\mu^2\mathrm{vec}(\Pb)^T\boldsymbol{\phi},
\end{align}
where
\begin{align}
\Pb=&\Gb[n]\Ab^T[n]\Cb_e\Ab[n]\Gb^T[n]\nonumber\\
&+\Hb[n]\sum_{j=1}^{K-1}{\Ab^T[n-j]\Cb_e\Ab[n-j]}\Hb[n].
\end{align}

It is proved in \cite{Loren16} that the above convergence satisfies if
\begin{equation}
0<\mu<\frac{2}{\lambda_{max}\left(\Bb_1\right)}
\end{equation}
where $\lambda_{max}(.)$ denotes the maximum eigenvalue of the matrix therein.

\end{appendices}


\begin{thebibliography}{1}


\bibitem{Sand13}
A. Sandryhaila, and J. M. F. Moura,
\newblock ``Discrete signal processing on Graphs,''
\newblock {\em  IEEE Transactions on Signal Processing}, vol. 61, no. 7, pp. 1644--1656, Apr. 2013.

\bibitem{Shum13}
D. I. Shuman, S. K. Narang, P. Frossard, A. Ortega, and P. Vandergheynst,
\newblock ``The emerging field of signal processing on Graphs: Extending high-dimensional data analysis to networks and other irregular domains,''
\newblock {\em IEEE Trans. Signal Processing Magazine}, vol. 30, no. 3, pp. 83--98, May 2013.


\bibitem{GSP18}
A. Ortega, P. Frossard, J. Kovacevic, J. F. Moura, and P. Vandergheynst,
\newblock ``Graph Signal Processing: Overview, Challanges, and Applications,''
\newblock {\em Proceedings of the IEEE}, vol. 106, pp. 808--828, 2018.


%



\bibitem{Chen15}
S. Chen, A. Sandryhaila, J. M. F. Moura, and J.~Kovacevic,
\newblock ``Signal Recovery on Graphs: Variation Minimization,''
\newblock {\em IEEE Transaction on Signal Processing.}, vol. 63, no. 17, pp. 4609--4624, Sep 2015.

\bibitem{Rome17}
D. Romero, M. Ma, and G. B.~Giannakis,
\newblock ``Kernel-based Reconstruction of Graph Signals,''
\newblock {\em IEEE Transaction on Signal Processing.}, vol. 65, no. 3, pp. 764--778, Feb 2017.

\bibitem{Berg17}
P. Berger, G. Hannak, and G.~Matz,
\newblock ``Graph Signal Recovery via Primal-Dual Algorithms for Total Variation Minimization,''
\newblock {\em IEEE Journal of Selected Topics in Signal Processing.}, vol. 11, no. 6, pp. 842--855, Sep 2017.

\bibitem{Qiu17}
K. Qiu, X. Mao, X. Shen, X. Wang, T. Li, and Y.~Gu,
\newblock ``Time-Varying Graph Signal Reconstruction,''
\newblock {\em IEEE Journal of Selected Topics in Signal Processing.}, vol. 11, no. 6, pp. 870--883, Sep 2017.




\bibitem{Wang15}
X. Wang, M. Wang, and Y.~Gu,
\newblock ``A Distributed Tracking Algorithm for Reconstruction of Graph Signals,''
\newblock {\em IEEE Journal of Selected Topics in Signal Processing.}, vol. 9, no. 4, pp. 728--740, June 2015.

%



\bibitem{Brug20}
E. Brugnoli, E. Toscano, and C.~Vetro,
\newblock ``Iterative reconstruction of signals on graph,''
\newblock {\em IEEE Signal Processing Letters.}, vol. 27, pp. 76--80, 2020.

\bibitem{Tork21}
R. Torkamani, and H. Zayyani,
\newblock ``Statistical Graph Signal Recovery Using Variational Bayes,''
\newblock {\em IEEE Trans. Circuit and Systems II: Express Briefs}, vol. 68, no. 6, pp. 2232--2236, June 2021.

\bibitem{Tork22}
R. Torkamani, H. Zayyani, and F. Marvasti
\newblock ``Joint Topology Learning and Graph Signal Recovery Using Variational Bayes in Non-Gaussian Noise,''
\newblock {\em IEEE Trans. Circuit and Systems II: Express Briefs}, vol. 69, no. 3, pp. 1887--1891, March 2022.





\bibitem{Loren16}
P. Di Lorenzo, S. Barbarossa, P. Banelli, and S.~Sardellitti,
\newblock ``Adaptive Least Mean Squares Estimation of Graph Signals,''
\newblock {\em IEEE Trans. on Signal and Inf. Proc. over Networks}, vol. 2, no. 4, pp. 555--568, Dec 2016.

\bibitem{Loren17}
P. Di Lorenzo, P. Banelli, S. Barbarossa, and S.~Sardellitti,
\newblock ``Distributed Adaptive Learning of Graph Signals,''
\newblock {\em IEEE Transaction on Signal Processing.}, vol. 65, no. 16, pp. 4193--4208, Aug 2017.

\bibitem{Loren18}
P. Di Lorenzo, P. Banelli, E. Isufi, S. Barbarossa,, and G.~Leus,
\newblock ``Adaptive Graph Signal Processing: Algorithms and Optimal Sampling Strategies,''
\newblock {\em IEEE Transaction on Signal Processing.}, vol. 66, no. 13, pp. 3584--3598, Jul 2018.

\bibitem{Loren18EUSIPCO}
P. Di Lorenzo, and E.~Ceci,
\newblock ``Online Recovery of Time- varying Signals Defined over Dynamic Graphs,''
\newblock {\em EUSIPCO 2018.}, Rome, Italy, Sep 2018.

\bibitem{Rame18}
M. Ramezani-Mayiami,
\newblock ``Joint Graph Learning and Signal Recovery via Kalman Filter for Multivariate Auto-Regressive Processes,''
\newblock {\em EUSIPCO 2018.}, Rome, Italy, Sep 2018.

\bibitem{Shen19}
Y. Shen, G. Leus, and G.~B~Giannakis,
\newblock ``Online Graph-Adaptive Learning with Scalability and Privacy,''
\newblock {\em IEEE Transaction on Signal Processing.}, vol. 67, no. 9, pp. 2471--2483, May 2019.

\bibitem{Loren19}
P. Di Lorenzo, S. Barbarossa, and S.~Sardellitti,
\newblock ``Distributed Adaptive Learning of Graph Processes via In-Network Subspace Projections,''
\newblock {\em Asilomar 2019.}, Pacific Grove, CA, USA, Nov 2019.

\bibitem{Spel20}
M. J. M. Spelta, and W.~A~Martins,
\newblock ``Normalized LMS algorithm and data-selective strategies for adaptive graph signal estimation,''
\newblock {\em Elsevier Signal Processing.}, vol. 167, pp. 2471--2483, Feb 2020.

\bibitem{Ahmadi20}
M. J. Ahmadi, R. Arablouei, and R~Abdolee,
\newblock ``Efficient Estimation of Graph Signals With Adaptive Sampling,''
\newblock {\em IEEE Transaction on Signal Processing.}, vol. 68, no. 9, pp. 3808--3823, 2020.

\bibitem{Zhao22}
Y. Zhao, and E. Ayanoglu,
\newblock ``Gaussian Kernel Variance for an Adaptive Learning Method on Signals Over Graphs,''
\newblock {\em IEEE Trans. on Signal and Information Proc. over Networks}, vol. 8, pp. 389--403, 2022.

\bibitem{Gogi21}
V. C. Gogineni, V. Naumova, S. Werner, and Y. F.~Huang,
\newblock ``Graph Kernel Recursive Least-Squares Algorithms,''
\newblock {\em APSIPA ASC 2021.}, Tokyo, Japan, Dec 2021.

\bibitem{Dutt00}
D. Duttweiler,
\newblock ``Proportionate normalized least-mean-squares adaptation in echo cancelers,''
\newblock {\em IEEE Trans. Speech Audio Process.}, vol. 8, no. 5, pp. 508--518, Sep 2000.

\bibitem{Wag11}
K. Wagner, and M.~Doroslovacki,
\newblock ``Proportionate-type normalized least mean square algorithms with gain allocation motivated by mean-square-error minimization for white input,''
\newblock {\em IEEE Trans. on Signal Proc.}, vol. 59, no. 5, pp. 2410--2415, May 2011.

\bibitem{Yim15}
S. H. Yim, H. S. Lee, and W. J.~Song,
\newblock ``A proportionate diffusion LMS algorithm for sparse distributed estimation,''
\newblock {\em IEEE Trans. on Circuit and Systems-II: Express Briefs.}, vol. 62, no. 10, pp. 992--996, Oct 2015.

\bibitem{ZayyJ21}
H.~Zayyani, et al, ``A robust generalized proportionate diffusion LMS
  algorithm for distributed estimation,'' \emph{IEEE Transactions on Circuits
  and Systems II: Express Briefs}, vol.~68, no.~4, pp. 1552--1556, 2020.

\bibitem{Shum13}
 D. I. Shuman, S. K. Narang, P. Frossard, A. Ortega, and P. Vandergheynst,
\newblock ``The emerging field of signal processing on graphs: Extending high-dimensional data analysis to networks and other irregular domains,''
\newblock {\em IEEE Signal Process. Mag.}, vol. 30, no. 3, pp. 83--98, May 2013.

\bibitem{Pesenson08}
I. Pesenson,
\newblock ``Sampling in Paley-Wiener spaces on combinatorial graphs,''
\newblock {\em Trans. Amer. Math. Soc.}, vol. 360, no. 10, pp. 5603--5627, 2008.

\bibitem{Mula18}
S. Mula, V. Ch. Gogineni, and  A. S. Dhar,
\newblock ``Algorithm and VLSI architecture design of proportionate-type LMS adaptive filters for sparse system identiﬁcation,''
\newblock {\em IEEE Trans. on Very Large Scale Integration (VLSI) Systems}, vol. 26, no. 9, pp. 1750--1762, 2018.

\bibitem{Pu21}
X. Pu, S. L. Chau, X. Dong, and D. Sejdinovic,
\newblock ``Kernel-based graph learning from smooth signals: A fundamental viewpoint,''
\newblock {\em IEEE Trans. on Signal and Inf. Process. oner Networks}, vol. 7, pp. 192--207, 2021.

\bibitem{Intel04}
P. Bodik, W. Hong, C. Guestrin, S. Madden, M. Paskin and R. Thibaux,
\newblock {Intel  lab  data},  [Online]. Available:http://db.csail.mit.edu/labdata/labdata.html., 2004.









\end{thebibliography}
\end{document}